\documentclass[11pt]{article}
\usepackage{amssymb}
\usepackage{amsmath}

\usepackage{graphicx}
\usepackage{enumerate}
\usepackage{hyperref}
\usepackage[nameinlink,capitalize,noabbrev]{cleveref} %
\usepackage{xcolor}
\usepackage[utf8]{inputenc}
\usepackage[ruled,lined,noend]{algorithm2e}
\usepackage{cite}

\ifthenelse{\isundefined{\definition}}{\newtheorem{definition}{Definition}}{}
\ifthenelse{\isundefined{\proof}}{\newenvironment{proof}{\noindent{\bf Proof.}\hspace{0.05in}}{\par\mbox{}\par}}{}
\ifthenelse{\isundefined{\theorem}}{\newtheorem{theorem}{Theorem}}{}
\ifthenelse{\isundefined{\lemma}}{\newtheorem{lemma}{Lemma}}{}
\ifthenelse{\isundefined{\corollary}}{\newtheorem{corollary}{Corollary}}{}
\ifthenelse{\isundefined{\qed}}{\newcommand{\qed}{\mbox{$\blacksquare$}}}{}
\ifthenelse{\isundefined{\property}}{\newtheorem{property}{Property}}{}
\ifthenelse{\isundefined{\example}}{\newtheorem{example}{Example}}{}

\newcommand{\NN}{\mathbb{N}}

\newcommand{\RR}{\mathbb{R}}

\newcommand{\TODO}[1]{} %
\newcommand{\Comment}[1]{} %

\newcommand{\SAppS}{\mathit{app}}
\newcommand{\SApp}[2]{\SAppS(#1, #2)}
\newcommand{\Struct}{\mathfrak{R}}

\newcommand{\TheoryRU}{\mathcal{RU}}

\newcommand{\VarAss}{\aleph}
\newcommand{\VarsF}[1]{\mathit{vars}(\phi)}

\newcommand{\Qed}{\qed}

\newcommand{\ModCont}[2]{\gamma^{#1}_{#2}}

\newcommand{\ExtendedVersion}{This is an extended and revised version of an article with the same name and author, which appeared in the Proceedings of the 48th International Symposium on Mathematical Foundations of Computer Science (MFCS 2023). }

\title{Deciding Predicate Logical Theories of\\Real-Valued Functions\footnote{\ExtendedVersion}}
\author{Stefan Ratschan\footnote{ORCID:0000-0003-1710-1513}\\Institute of Computer Science, Czech Academy of Sciences}

\begin{document}

\maketitle

\begin{abstract}
  The notion of a real-valued function is central to mathematics, computer science, and many other scientific fields. Despite this importance, there are hardly any positive results on decision procedures for predicate logical theories that reason about real-valued functions. This paper defines a first-order predicate language for reasoning about multi-dimensional smooth real-valued functions and their derivatives, and demonstrates that---despite the obvious undecidability barriers---certain positive decidability results for such a language are indeed possible: The quantifier-free case is decidable, since it allows the reduction of reasoning about real-valued functions to the reasoning about real numbers and uninterpreted function symbols. In the undecidable case with quantifiers ranging over real-valued variables, it is possible to algorithmically detect satisfiability for all formulas possessing certain robustness properties.

\end{abstract}

\section{Introduction}

Predicate logical decision procedures have become a major workhorse in computer science, for example, as the basic reasoning engines in SAT modulo theory (SMT) solvers~\cite{Barrett:18}. Common decision procedures support theories such as uninterpreted function symbols, arrays, linear integer arithmetic, and real arithmetic. However, many areas of computer science (e.g., computer aided design, formal verification of physical systems, machine learning) use as their basic data structure not only real numbers but real-valued \emph{functions}, for example, to represent solid objects~\cite{Farin:88}, correctness certificates~\cite{Prajna:04,Platzer:18} or neural networks~\cite{Aggarwal:18}. Moreover, due to their fundamental role as basic mathematical objects, real-valued functions are used as a basic modeling tool throughout many further scientific areas. But unfortunately, real-valued functions have been left almost completely untouched by research on predicate logical decision procedures. The goal of this paper is to take a first step to fill this gap. \pagebreak[2]

More concretely, the paper provides the following contributions, each corresponding to one of the subsequent sections starting from Section~\ref{sec:lang}:
\begin{itemize}
\item We formalize a first-order language of real-valued functions that allows reasoning about both real numbers and multi-dimensional real-valued smooth functions based on the usual arithmetical operations, function evaluation and differentiation.
\item We prove that a quantifier-free fragment of the language that restricts arithmetic to addition and multiplication of real numbers, but still provides function evaluation and differentiation, is decidable. 
\item We prove that for a fragment of the language that keeps the restriction of arithmetic to addition and multiplication of real numbers but allows arbitrary quantification on real-valued variables (but not on function-valued variables), there is an algorithm that can detect satisfiability for all input formulas that are robustly satisfiable in the sense that there is a satisfying assignment that stays satisfying under small perturbations of the values of  function-valued variables.
\item We prove that removing the restriction of arithmetic to addition and multiplication of real numbers, allowing transcendental function symbols such as $\sin$ and $\cos$ (in general, real-valued functions that are computable in the sense of computable analysis~\cite{Pour-El:89,Brattka:21})
  allows a similar algorithm whose success depends on a different notion of robustness. This generalizes previous results~\cite{Gao:12,Ratschan:02f} to the case with real-valued functions.
\item We combine these results to take into account formulas whose individual parts belong to different sub-classes from the previous items.
\end{itemize}
These results provide a straightforward methodology for extending existence results for functional constraints to computability results. For example, we demonstrate that despite the fact that polynomial ordinary differential equations are able to simulate Turing machines~\cite{Brattka:21}, it is possible to extend a result showing the existence of polynomial Lyapunov functions for asymptotically stable ordinary differential equations~\cite{Peet:09} to semi-decidability of asymptotic stability for polynomial ordinary differential equations.

We neither claim theoretical nor practical efficiency of the resulting decision procedures. Instead, our goal is to overcome scientific fragmentation by developing a  framework that can be instantiated to more efficient techniques for specific applications.

The paper has the following structure: In the next section, we discuss related work. In Section~\ref{sec:lang}, we define the syntax and semantics of the mentioned predicate language for reasoning about smooth real-valued functions. In Section~\ref{sec:quantifier-free-case} we prove decidability of the quantifier- and transcendental function free case. In Section~\ref{sec:scal-quant} we discuss decidability of the case with arbitrary quantification on real-valued variables, but still without function symbols denoting transcendental function. In Section~\ref{sec:transcendental}, we extend these result for the case of function symbols denoting transcendental function. In Section~\ref{sec:combination}, we combine the results from the previous sections. In Section~\ref{sec:discussion} we discuss practical consequences of the results. And in Section~\ref{sec:conclusion} we conclude the paper.

\section{Related Work}

Reasoning about real-valued functions---that we also simply call \emph{real functions}---will, of course, be usually based on reasoning about real numbers. This is facilitated by the fact that unlike the case of the integers, in the case of the real numbers, its non-linear theory (i.e., the theory of real closed fields) is decidable~\cite{Tarski:51}. The decidability of the case with the exponential function is still unknown, but is decidable provided Schanuel's conjecture holds~\cite{MacIntyre:96}. Inclusion of the sine function makes the problem undecidable since---as a periodic function---it is able to encode the integers. This makes any theory that allows reasoning about systems of linear ordinary differential equations (ODEs) undecidable, since the sine function appears as the solution of the linear ODE $\dot{x}=-y, \dot{y}=x$.

In mathematical analysis, real functions are often abstracted to elements of abstract function spaces such as Banach spaces and Hilbert spaces~\cite{Lax:02}. However, with one notable exception~\cite{Solovay:12} we are aware of, corresponding predicate logical decision problems have been largely ignored by computer science.

An important occurrence of real functions is in the role of solutions of ordinary differential equations (ODEs) and hybrid dynamical systems. Formal verification of such systems has been an active research topic over many years~\cite{Doyen:18}, with a plethora of decidability and undecidability results~\cite{Fraenzle:99,Collins:09,Bournez:08,Bournez:18}. 
Deductive verification bases formal verification on automated reasoning frameworks such as hybrid dynamic logic~\cite{Platzer:18}, or proof assistants such as Isabelle/HOL~\cite{Foster:21}. Reasoning with functions as the solution of ODEs has been included into SAT solvers without formulation as a first-order decision problem~\cite{Eggers:08,Gao:13}, or using an alternative semantics based on floating-point approximation~\cite{Kolarik:20a}.

ODEs have also played a role as objects in constraint programming~\cite{Hickey:00}. In contrast to the work mentioned in this paragraph, in this paper, we introduce a general logical language with variables and predicate and function symbols ranging over real-valued functions. Especially, we allow multi-dimensional functions and partial differentiation, whereas ODEs and hybrid systems are defined using one-dimensional functions, only (each component of the state corresponding to a function whose single dimension is time).

Computation in function spaces plays a major role in numerical analysis, where it is mostly restricted to representing solutions to certain specific computation problems, especially, solving ordinary or partial differential equations. There are also some general approaches to computing with functions~\cite{Driscoll:14,Collins:11a}. However, the basic assumption in numerical analysis is that the solution to the given problem exists and is unique, and the goal is to compute an approximation of this solution, whereas in this paper we consider satisfiability questions, where a proof of existence is the goal, not an assumption.

Computer algebra~\cite{Gathen:03} studies computation with symbolic objects, especially polynomials, that can be interpreted as representations of real functions. Corresponding algorithms exist for solving differential equations~\cite{Winkler:19} or functional equations~\cite{Brown:24}, but they are either restricted to very specific classes of equations, or do not provide any guarantees for success, being heuristic in nature. Unlike that, in this paper we are interested in solving problems of reasoning about functions that are independent from a certain representation, which will allow us to provide formal guarantees of success for very general classes of problems.

The proof of decidability of the quantifier-free case will be based on abstracting function variables to uninterpreted function symbols. Abstraction to uninterpreted function symbols is a classical technique in formal verification~\cite{Bradley:07} that has also been applied to real functions~\cite{Cimatti:18}, but with the goal of modeling \emph{specific} function symbols, while in this paper we are interested in general reasoning about smooth real functions and their derivatives.

For quite some time, robustness has been recognized as tool for characterizing solvable cases of undecidable decision problems. It was used for dynamical systems~\cite{Fraenzle:99,Asarin:01} and for decision procedures for the real numbers~\cite{Ratschan:01c,Gao:12}. However, all of those results do not allow a general language for reasoning about real functions.

\section{Formal Syntax and Semantics}
\label{sec:lang}

In this section, we define the syntax and semantics of the first-order language for reasoning about real functions that we will want to decide. As a first example, consider the formula \[ \exists X\; \forall u,v \;.\; \SApp{\partial_1 X}{u,v}=1\wedge \SApp{\partial_2 X}{u,v}\leq u^2,\] that asks the question whether there exists a smooth function in $\mathbb{R}^2\rightarrow\mathbb{R}$ whose partial derivative in its first argument is one everywhere, and whose partial dervative in its second argument is less or equal the square of its first argument. The reader will find more examples at the beginning of each of the two following sections.

The language will be sorted, allowing variables that range over real numbers and variables that range over real functions.
We denote the sort ranging over real numbers, that we also call the \emph{scalar sort}, by $\mathcal{R}$, and the sorts ranging over real-valued functions, that we also call the \emph{function sorts}, by 
$\mathcal{F}_n$, where $n\in\mathbb{N}$ refers to the number of arguments (i.e., dimension of the domain). We will also use the symbol $\mathcal{F}$ to stand for any sort $\mathcal{F}_i,i\in\NN$. For each of those sorts, we assume a corresponding set of variables $\mathcal{V}=\mathcal{V}_{\mathcal{R}}\cup\mathcal{V}_{ \mathcal{F}_1}\cup\mathcal{V}_{ \mathcal{F}_2}\dots$. We will write the elements of $\mathcal{V}_{\mathcal{R}}$ using lowercase letters and call them \emph{scalar variables}. We will write the elements of $\mathcal{V}_{ \mathcal{F}_1}\cup\mathcal{V}_{ \mathcal{F}_2}\dots$ using uppercase letters and call them \emph{function variables}. We will also use the symbol~$\mathcal{V}_{ \mathcal{F}}$ to denote the set of all function variables $\mathcal{V}_{ \mathcal{F}_1}\cup\mathcal{V}_{ \mathcal{F}_2}\dots$.

We will build formulas based on the usual syntax of many-sorted first-order logic. Here, we allow rational constants, arithmetical function symbols such as $+, \times, \exp, \sin$, and predicate symbols $ =, \leq, \geq, <, >$ of the usual arity, all with arguments and---in the case of function symbols---also results of the sort $\mathcal{R}$.  For every $n\in\mathbb{N}$,  we allow the function symbols $\SAppS: \mathcal{F}_n\times\mathcal{R}^n\rightarrow\mathcal{R}$ and $\partial_i: \mathcal{F}_n\rightarrow \mathcal{F}_n$, $i\in \{1,\dots,n\}$ that we call app-operator and differentiation operator, respectively. As usual, we will often write the differentiation operator without parenthesis, and for  $X\in \mathcal{F}_1$, we also write $\dot{X}$ instead of $\partial_1X$. We will also call a term whose outermost symbol is the function symbol $\SAppS$, an \emph{app-term}.
We will call formulas whose function symbols are restricted to $\{ +, \times,\SAppS \} \cup \{ \partial_i\mid i\in\mathbb{N} \}$, and hence avoiding transcendental function symbols, \emph{function-algebraic}.

Throughout the paper, we will use the usual abbreviations $\forall x\in [\underline{a},\overline{a}]\;.\; \phi$ for $\forall x \;.\; [ \underline{a}\leq x \wedge x\leq \overline{a}]\Rightarrow \phi$, and $\exists x\in [\underline{a},\overline{a}]\;.\; \phi$ for
$\exists x \;.\;  \underline{a}\leq x \wedge x\leq \overline{a}\wedge \phi$. We will call a formula that only contains quantifiers that are bounded in this form, with $\underline{a}$ and $\overline{a}$ being rational numbers, a \emph{bounded formula}.

We define the semantics of formulas by defining a structure $\Struct$ giving the usual real-valued semantics to all function and predicate symbols. This allows us to avoid  questions of axiomatization and, at the same time, ensures compatibility with the common intuition. Clearly, satisfiability of a formula based on classical mathematical semantics implies its satisfiability wrt. an arbitrary axiomatization compatible with classical mathematics. 

In more detail, the structure $\Struct$ will be many-sorted, where the sort $\mathcal{R}$ ranges over the real numbers $\mathbb{R}$ and the sorts $\mathcal{F}_n$, $n\in\mathbb{N}$ range over the set of smooth (i.e., infinitely often differentiable) functions in $\mathbb{R}^n\rightarrow \mathbb{R}$. We will use the usual lambda-notation to denote such functions. For example, $\lambda x\:.\: x$ denotes the identity function. We will use the notation that for any smooth function $F: \RR^n\rightarrow\RR$, and tuple $(\beta_1,\dots, \beta_n)\in\NN_0^n$, $D^{(\beta_1,\dots, \beta_n)} F$ denotes the function that is the result of differentiating for every $i\in \{1,\dots,n\}$ the function $F$ $\beta_i$-times wrt. its $i$-th component.
The semantics of function and predicate symbols on the real numbers will be as usual.
The app-operator and differentiation operator are defined as follows:
\begin{itemize}
\item For every  $n\in\mathbb{N}$, for all $X\in \mathcal{F}_n$, $\Struct(\SAppS)(X,x_1,\dots,x_n):=X(x_1,\dots,x_n)$ (i.e., function application in its usual mathematical sense)

\item For every $n\in\mathbb{N}$,  $i\in \{1,\dots,n\}$, for every $X\in \mathcal{F}_n$, $\Struct(\partial_i)(X):=D^d X$, where $d\in\NN_0^n$ with $d(i)=1$ and for every $k\neq i$, $d(k)= 0$ (i.e., the result of taking the derivative of $X$ wrt. its $i$-th argument).
\end{itemize}

We will denote the set of variable assignments assigning to each variable an element of its respective domain, by $\VarAss$. We can now assign semantics to formulas in the usual way:
Given an assignment $\alpha\in\VarAss$, we can extend it to terms, writing $\alpha(t)$ for the value of the term~$t$ in the structure $\Struct$ under the assignment $\alpha$. We will write $\alpha\models\phi$ iff the interpretation given by structure~$\Struct$ and assignment~$\alpha$ satisfies $\phi$.  We call a formula $\phi$ \emph{satisfiable} iff there is an  assignment $\alpha\in\VarAss$ such that $\alpha\models \phi$. In such a case we will also say that $\phi$ is $\mathcal{F}$-satisfiable. By abuse of notation, we will use the symbol $\mathcal{F}$ to not only denote the function sorts, but also the 
\emph{theory} $\mathcal{F}$ of $\mathcal{F}$-satisfiable formulas.

\section{Quantifier-Free Case}
\label{sec:quantifier-free-case}

In this section, we consider formulas that are quantifier-free and function-algebraic. Here are some examples:
\begin{itemize}
\item $\SApp{X}{t}\geq 1 \wedge \SApp{X}{t+1}^2\leq 1$: This formula restricts the value of the function $X$ at two different points $t$ and $t+1$. Since these points are different, for checking satisfiability of the formula, it suffices check satisfiability of the algebraic inequalities $r\geq 1 \wedge s^2\leq 1$. Based on a satisfying assignment for this formula, we get a satisfying assignment for the original formula by assigning to $X$ a function interpolating between the values for $r$ and~$s$.
\item $\SApp{X}{0}=0\wedge \SApp{\dot{X}}{1}=\SApp{X}{1}^2$. This formula not only restricts values for the function $X$, but also states a relationship between the value of $X$ and its derivative. The formula is satisfiable since the identity function satisfies the properties stated by the formula.
\item $\SApp{\partial_1 X}{t}=1\wedge \SApp{\partial_2 X}{t}=1$. This formula states a relationship between two partial derivatives of $X$ at the same value $t$. This holds, for example, for the function $X$ with $X(u,v)=u+v$.
\end{itemize}

The basic idea for deciding such formulas is, that quantifier-free formulas constrain the values of function variables only at a finite (but not fixed) subset of their domain which will allow us to treat them as uninterpreted function symbols. To do so, we have to get rid of the app- and differentiation operators. For this, observe that the only syntactic elements that result in terms of function sort are function variables and differentiation operators. Hence, differentiation operators can only occur in the form of terms of the form $\partial_{i_1} \partial_{i_2} \dots \partial_{i_n} V$. So we let $\tau_{\partial}(\phi)$ be the formula resulting from replacing every maximal term of this form (i.e., every term of this form that is not an argument of a differentiation operator) by a fresh function variable $V_{i_1,\dots,i_n}$. For example,
\[\tau_{\partial}(\SApp{\partial_1 X}{t}+\SApp{\partial_2 X}{t}=1)\equiv\SApp{X_1}{t}+ \SApp{X_2}{t}=1.\]

The next step is to get rid of the app operator. For this,
we denote, for every quantifier-free formula~$\phi$, by~$\tau(\phi)$ the result of
replacing every app-term $\SApp{X}{t_1,\dots, t_k}$ of $\tau_{\partial}(\phi)$ by $X(t_1,\dots, t_k)$ where in the resulting formula, we now consider $X$ a $k$-ary function symbol. Continuing the example, we get \[\tau(\SApp{\partial_1 X}{t}+\SApp{\partial_2 X}{t}=1) \equiv X_1(t)+ X_2(t)=1.\]

The resulting formula is a formula in the language defined by the combination of the signature of the theory of real-closed fields and the signature of the theory of uninterpreted function symbols. 
The combination of these two theories, that we will denote by $\TheoryRU$, is decidable: The signatures of the theory of uninterpreted function symbols and the theory of real-closed fields only share equality, and both are stably infinite\footnote{A theory $T$ with signature $\Sigma$ is called stably infinite iff for every quantifier-free
$\Sigma$-formula $\phi$, if $F$ is $T$-satisfiable, then there exists some $T$-interpretation
that satisfies $F$ and has a domain of infinite cardinality~\cite{Nelson:79,Bradley:07}.}, hence the decision procedures for the individual procedures can be combined to a decision procedure for the combined theory by the Nelson-Oppen theory combination procedure~\cite{Nelson:79,Bradley:07}. As a consequence, we can algorithmically decide $\TheoryRU$-satisfiability \footnote{A formula $\phi$ is $\TheoryRU$-satisfiable iff there is an interpretation $I$ satisfying both the axioms of $\TheoryRU$ and $\phi$.} of translated formulas~$\tau(\phi)$.
Moreover, the translation preserves satisfiability:
\begin{theorem}
\label{thm:equisat}
  A quantifier-free conjunctive formula $\phi$ is $\mathcal{F}$-satisfiable if and only if 
$\tau(\phi)$ is $\TheoryRU$-satisfiable.
\end{theorem}

For proving this theorem, we have to bridge two differences between $\mathcal{F}$- and $\TheoryRU$-satisfiability: First, the semantics of $\mathcal{F}$-satisfiability restricts the domain of function variables to specific functions, more concretely, to smooth real function.  And second, the theories of real closed fields and uninterpreted function symbols are defined using axioms, unlike our theory~$\mathcal{F}$ that we defined semantically, by fixing a certain structure. Before going into the details of the proof, we state a few lemmas. The first one extracts the non-algorithmic core of the Nelson-Oppen method~\cite{Nelson:79,Tinelli:96,Bradley:07}:

\begin{lemma}
  \label{lem:nelson_oppen}
  Let $T_1$ and $T_2$ be two stably infinite theories of respective signatures $\Sigma_1$ and $\Sigma_2$, having only equality in common. Let $\phi_1$ be a conjunctive $\Sigma_1$-formula, and $\phi_2$ a conjunctive $\Sigma_2$-formula. Then $\phi_1\wedge\phi_2$ is 
$(T_1\cup T_2)$-satisfiable   iff there is an
equivalence relation $E$ on the common variables $V:= var(\phi_1)\cap var(\phi_2)$ s.t. $\phi_1\wedge\rho(V, E)$ is $T_1$-satisfiable and $\rho(V,E)\wedge\phi_2$ is $T_2$-satisfiable, where $\rho(V, E)$ is the formula \[\bigwedge_{u,v\in V\;.\; uEv} u=v \wedge \bigwedge_{u,v\in V\;.\; \neg (uEv)} u\neq v.\]
\end{lemma}

Every $(\Sigma_1\cup \Sigma_2)$-formula $\phi$ can be brought into an equi-satisfiable formula of the form $\phi_1\wedge\phi_2$, where $\phi_1$ is a $\Sigma_1$-formula, and  $\phi_2$ is a $\Sigma_2$-formula using the so-called variable abstraction phase of the Nelson-Oppen method. In our case,  $T_1$ is the theory of real closed fields, and $T_2$ the theory of uninterpreted function symbols. For the result  \[X(t)\geq 1 \wedge X(t+1)^2\leq 1,\] of translating the first example from the beginning of the section, the result of the variable abstraction phase is the equi-satisfiable formula
\[v_1\geq 1 \wedge v_2=t+1\wedge v_3^2\leq 1 \wedge v_1=X(t) \wedge v_3=X(v_2).\]
The common variables are $\{ v_1,v_2,v_3, t\}$, and  the equivalence relation induced by the set of equivalence classes~$\{\{ v_1,v_3\}, \{v_2\}, \{ t\}\}$ illustrates Lemma~\ref{lem:nelson_oppen}, since \[v_1\geq 1 \wedge v_2=t+1\wedge v_3^2\leq 1 \wedge v_1=v_3\wedge v_1\neq v_2\wedge v_3\neq v_2\]
is satisfiable in the theory of real-closed fields, and
 \[v_1=X(t) \wedge v_3=X(v_2) \wedge v_1=v_3\wedge v_1\neq v_2\wedge v_3\neq v_2\]
is satisfiable in the theory of uninterpreted function symbols.

The second lemma states a Hermite-like interpolation property whose proof follows from standard techniques in mathematical analysis. 
\begin{lemma}
  \label{lem:interpolation}
  Let $p$ be a function from a finite subset $P$ of $\RR^n\times\NN_0^n$ to $\RR$. Then there exists a smooth function $F:\RR^n\rightarrow\RR$ s.t. for every $(x,d)\in P$, $(D^d F)(x)=p(x,d)$.
\end{lemma}

\begin{proof}
  Let $X$ be the set $\{ x \mid (x,d)\in P \}$. This set if finite, and hence the elements of $X$ are isolated. For each  $c\in X$, construct a smooth function~$f_c$ that for all $d$ with $(c,d)\in P$,  $(D^{d} f_c)(c)=p(c,d)$. Let $F:\RR^n\rightarrow\RR$ be such that for all $x\in\RR^n$, $F(x)=\sum_{c\in X} B_c(f_c(x))$, where $B_c$ is a smooth function that is equal to the identity function in a sufficiently small neighborhood of $c$, and zero around all other elements of $X$ (i.e., a so-called bump function).  Then $F$ satisfies the desired property. \Qed
\end{proof}

Now we return to the proof of Theorem~\ref{thm:equisat}:

\begin{proof}
  To prove the $\Rightarrow$ direction,
  we assume a variable assignment $\alpha$ that $\mathcal{F}$-satisfies $\phi$ and construct an interpretation that satisfies both the axioms of $\TheoryRU$ and the formula $\tau(\phi)$. The interpretation is based on the structure of the real numbers, interprets the  symbols $\{0, 1, +,\times, \leq, <, \geq, >  \}$ in the usual mathematical way, interprets the function symbols introduced by the translation $\tau$ as the corresponding smooth real-valued functions given by $\alpha$ and their respective derivatives, and assigns to the variables of $\tau(\phi)$ the corresponding real values given by~$\alpha$.
The result satisfies the formula~$\phi$ by construction and satisfies the axioms of $\TheoryRU$ since the real numbers are an instance of the theory of real closed fields, 
  
We are left with proving the $\Leftarrow$ direction. For this, we assume that $\tau(\phi)$ is $\TheoryRU$-satisfiable, and build an assignment that satisfies $\phi$, assigning to each variable of $\phi$ an element of the domain of its respective sort (i.e., either a real number or a smooth real function).

Observe that both the theory of real closed fields and the theory of uninterpreted function symbols are stably infinite. Hence we can apply the Nelson-Oppen method.
 Let the formulas $\pi_R$ and $\pi_U$ be the result of applying the variable abstraction phase of the Nelson-Oppen method to $\tau(\phi)$. Hence, $\pi_R$ is a formula in the language of real closed fields, and $\pi_U$ a formula in the language of uninterpreted function symbols  s.t. $\pi_R\wedge\pi_U$ is $\TheoryRU$-equi-satisfiable with $\tau(\phi)$.
 Let $V$ be the common variables of $\pi_R$ and $\pi_U$, and let $E$ be the equivalence relation on $V$ ensured by Lemma~\ref{lem:nelson_oppen}.  Then  $\pi_R\wedge\rho(V,E)$ is satisfiable in the theory of real closed fields and $\rho(V,E)\wedge\pi_U$ in the theory of uninterpreted function symbols. Note that those theories are defined axiomatically, and hence  satisfying interpretations do not necessarily have to be based on real numbers.

 The theory of real closed fields is complete, and hence all its models are elementary equivalent (i.e., all models satisfy the same sentences in the language of real closed fields). 
 As a consequence, there is an interpretation~$I_R$ that satisfies $\pi_R\wedge\rho(V,E)$ and assigns  real numbers to all variables. 

Also $\rho(V,E)\wedge\pi_U$ has a satisfying interpretation. However, since the theory of uninterpreted symbols is not complete, we need a more involved construction to come up with an interpretation assigning real numbers and real-valued functions.

We observe that for a formula that is satisfiable in the theory of uninterpreted function symbols, the congruence closure algorithm~\cite{Nelson:80,Bradley:07} constructs a satisfying interpretation. Its domain is formed by  equivalence classes $T_{\sim}$ of the set of sub-terms~$T$ of the given formula. This domain is finite since the set $T$ is finite. Moreover, each equivalence class contains only finitely many terms. Let $I_U$ be such an interpretation  satisfying  the formula~$\rho(V,E)\wedge\pi_U$. Observe that the arrangement $\rho(V,E)$ ensures that all variables shared by $\pi_R$ and $\pi_U$ belong to the same equivalence class iff they have the same value in $I_R$.

We will combine the interpretations $I_R$ and $I_U$ into an interpretation~$I$ that  satisfies $\tau(\phi)$ and, in addition, uses the real numbers as its domain.
Hence, we will translate the elements in the domain of $I_U$ to real numbers, and extend them to real-valued functions corresponding to the uninterpreted function symbols.

We will now define a function $r$ assigning to each 
equivalence class in $T_{\sim}$ a distinct real number. 
Let this function $r$ be such that it assigns
to each equivalence class containing a variable from $\pi_R$ the value of this variable in $I_R$ (as we have observed, this is unique over all such variables belonging to the same equivalence class), and to all other equivalence classes a further, distinct real number.
Let $r': T\rightarrow\mathbb{R}$ s.t. for every term $t\in T$, $r'(t)$ is the real number that $r$ assigns to the equivalence class containing $t$.

Based on this, let $I$ be the following interpretation, which assigns real numbers to all variables in $\tau(\phi)$ and partial real functions to the function symbols in $\tau(\phi)$:
\begin{itemize}
\item for every variable $x$ occurring in $\pi_R$, $I(x):= I_R(x)$,
\item for every variable $x$ occurring in $\pi_U$, $I(x):= r'(x)$,
\item for every function symbol $X$ of arity~$k$, let $I(X)$ be the partial function such that for every term of the form $X(t_1,\dots, t_k)$ in  $T$,
  \[I(X)(r'(t_1),\dots, r'(t_k)):= r'(X(t_1,\dots,t_k)),\] and $I(X)$ is undefined for all other values.
\end{itemize}
The two following observations make this definition well-formed:
\begin{itemize}
\item The first two items overlap. This is no problem since for shared variables, $I(x)$ and $r'(x)$ coincide. 
\item The definition in the third item is unique since due to the congruence axioms of the theory of free function symbols, for all $t_1,\dots, t_k,t_1'\dots,t_k'\in T$,
  $r'(t_1)=r'(t_1'), \dots, r'(t_k)=r'(t_k')$ implies \[r'(X(t_1,\dots,t_k))=r'(X(t_1',\dots,t_k')).\]
\end{itemize}

We will now build a variable assignment~$\alpha$ from $I$ such that $\alpha\models\phi$. For every scalar variable $x$, $\alpha(x):=I(x)$. For every function variable $V$, $\alpha(V)$ will be a smooth real-valued function whose values coincide with the values of the partial function $I(V)$ on all points where this partial function has a defined value, and whose derivatives coincides with the values of the corresponding partial function~$I(V_{i_1,\dots,i_n})$ on all points where the latter has a defined value. Such a function exists due to Lemma~\ref{lem:interpolation}, and it satisfies the formula $\phi$. \Qed
\end{proof}

To illustrate the theorem, we continue with the example from above. The real part of the formula is satisfiable, for example by  $\{ v_1\mapsto 1, v_2\mapsto 7, v_3\mapsto 1, t\mapsto 6 \}$. Applying the congruence closure algorithm to the part with uninterpreted function symbols, we work with the set of sub-terms $T=\{ v_1, v_2, v_3, t, X(t), X(v_2)\}$. The result of the congruence closure algorithm is the equivalence relation $\{\{ v_1, v_3, X(v_2), X(t) \}, \{t\}, \{ v_2 \}\}$. Hence $r'$ is  $\{ v_1\mapsto 1, v_3\mapsto 1, X(v_2)\mapsto 1, X(t)\mapsto 1, t\mapsto 6,v_2\mapsto 7\}$. Since every variable in $\pi_U$ also occurs in $\pi_R$, the interpretation $I$ can simply agree with $r'$
on the scalar variables. Moreover, it assigns to the function variable~$X$ the partial function $\{ 7\mapsto 1, 6\mapsto 1\}$. The corresponding assignment $\alpha$ assigns to the scalar variables the same real values as $I$, and assigns to $X$ a smooth interpolation of the  partial function $\{ 7\mapsto 1, 6\mapsto 1\}$. For example, this could be the constant function that has the value $1$, everywhere.

Since a disjunction of formulas is satisfiable, if one of the constituting disjuncts is satisfiable, we get:

\begin{corollary}
  \label{cor:decidable}
  The quantifier free, function-algebraic theory of real functions is decidable.
\end{corollary}

The proof of Theorem~\ref{thm:equisat} also shows  how to compute satisfying assignments: After checking the satisfiability of $\tau(\phi)$ using the congruence closure algorithm and the Nelson-Oppen combination procedure,  construct the variable assignment $\alpha$ defined in the proof.

Note that the concluding building block of the proof of Theorem~\ref{thm:equisat} is Lemma~\ref{lem:interpolation}. Any analogous lemma that ensures stronger properties of the constructed functions results in a corresponding strengthening of Theorem~\ref{thm:equisat}. For example, we could also be interested in constructing functions that generalize the constraints given by the input formula as much as possible, maximizing certain regularity properties~\cite{Fefferman:20}.

\section{Scalar Quantification}
\label{sec:scal-quant}

We will now allow arbitrary quantification on scalar variables, but not on function variables. We will still require formulas to be function-algebraic---a restriction that we will remove in the next section.
An example of a formula with quantification on scalar variables is
  \[      \forall t \;\exists t' \;.\; \SApp{X}{t,2t}^2+1\geq \SApp{\partial_2 X}{t',t'},\]
  where $X\in\mathcal{F}_2$.  This formula is satisfiable by the assignment $\{ X\mapsto \lambda x\:.\: 0\}$.

  Many problems resulting from the synthesis of correctness certificates for continuous systems (e.g., Lyapunov function~\cite{Khalil:02}, barrier certificates~\cite{Prajna:04} and their generalizations~\cite{Platzer:18,Han:21,Edwards:24}), belong to the class under discussion. It allows the formulation of initial value problems such as
\[
  \SApp{X}{0}=1\wedge\forall t\;.\; t\geq 0 \Rightarrow \SApp{\dot{X}}{t}=\SApp{X}{t}^2,
  \]
  which is unsatisfiable, since the solution of this initial value problem grows over all bounds as $t$ approaches $1$. In general, the existence of (time-unbounded) solutions of polynomial initial value problems is undecidable~\cite{Graca:08a}, which shows that the 
satisfiability of the class formulas under discussion is undecidable. However, we will see that it is still possible to provide a certain characterization of computability based on a robustness requirement on the input formula.

In Subsection~\ref{sec:prov-satisf}, we will introduce a method for checking satisfiability under the condition that the functional variables are instantiated to fixed, user-provided polynomials. In Subsection~\ref{sec:search-algorithm} we will introduce an algorithm for checking satisfiability that systematically searches for such instantiations.
In Subsection~\ref{sec:robust-satisf}, we will introduce a robustness property of formulas that will allow us to characterize termination of this algorithm. And in Subsection~\ref{sec:completeness}, we will show how the introduced robustness property
ensures success of the algorithm introduced in Subsection~\ref{sec:search-algorithm}.

\subsection{Checking Satisfiability Under Polynomial Instantations}
\label{sec:prov-satisf}

The method for checking satisfiability of a formula that we introduce in this subsection depends on an instantation of its functional variables to polynomials with rational coefficients. This will allow us to rewrite the formula into a formula in the language of the theory of real-closed fields which is decidable~\cite{Tarski:51}. In this subsection, we still assume these polynomials to be given (e.g., chosen by the user), and drop this assumption later.

\begin{definition}
  \label{def:poly_eval}
  Let $V$ be a finite set of function variables. We call a function~$\pi$ that assigns to every function variable in $V$ with sort $\mathcal{F}_n$  a  polynomial  in the variables $v_1,\dots,v_n$ with rational coefficients a \emph{polynomial assignment} on $V$. Moreover, we call a pair consisting of a formula with function variables from a finite set $V$ and a polynomial assignment on $V$ an \emph{instantiated formula}. 
\end{definition}

A polynomial assignment $\pi$ in an instantiated formula $(\phi, \pi)$ instantiates each function variable in $\phi$ to the respective polynomial assigned by $\pi$, but does not assign values to scalar variables. Since the function variables have values that are polynomials with rational coefficients, a polynomial assignment allows us to evaluate terms symbolically. For a given instantiated formula~$(\phi, \pi)$, we first replace all function variable $X$ in the formula~$\phi$ by the corresponding polynomial~$\pi(X)$. Then we simplify the resulting formula using the usual rules, rewriting
\begin{itemize}
\item $\partial_i c$ to $0$, for a constant $c$,
\item $\partial_i v_j$ to $1$, if $i=j$, and to $0$, otherwise,
\item $\partial_i (T_1+T_2)$ to $(\partial T_1)+(\partial T_2)$,
\item $\partial_i (T_1T_2)$ to $(\partial T_1)T_2 + T_1(\partial T_2)$, and
\item $\SApp{T, t_1,\dots,t_n}$ to the result of substituting $t_1$ for $v_1$, $\dots$, $t_n$ for $v_n$ in $T$  
\end{itemize}

For example, a symbolic evaluation of the term \[\SApp{X}{q+1}+p^2\SApp{\partial_1 X}{r}\SApp{Y}{q}\] for the polynomial assignment $\{ X\mapsto v^2+1, Y\mapsto v \}$ looks as follows:
\newcommand{\Poly}[1]{#1}
\[ \begin{array}{l}
\SApp{\Poly{v^2+1}}{q+1}+p^2\SApp{\partial_1 \Poly{v^2+1}}{r}\SApp{\Poly{v}}{q} \\
\SApp{\Poly{v^2+1}}{q+1}+p^2\SApp{\Poly{2v}}{r}\SApp{v}{q}\\
\Poly{(q+1)^2+1+p^22rq}
  \end{array} \]

The result is a polynomial in the scalar variables of the term. Now denote by $\Pi_{\pi}(\phi)$ the result of applying symbolic term evaluation to all terms in the formula $\phi$. For example,
\[\Pi_{\pi}(\forall p\forall q\;.\; \SApp{X}{q+1}+p^2\SApp{\partial_1 X}{r}\SApp{Y}{q}\geq 0)\] with $\pi$ being the polynomial assignment from above, is
\[ \forall p\forall q\;.\; \Poly{(q^2+1)+1+p^22rq}\geq 0.\]

Polynomial evaluation completely eliminates any function variables or operators, and hence for every instantiated formula $(\phi, \pi)$, $\Pi_\pi(\phi)$ neither contains a function variable, nor an app- or diff-operator.  Therefore,  $\Pi_{\pi}(\phi)$ is a formula in the language of real-closed fields, which is decidable.

Moreover, instantiated formulas can be used for proving satisfiability:

\begin{property}
  \label{prop:instantiation}
  Let $(\phi, \pi)$ be an instantiated formula.
  Then  $\phi$ is satisfiable by a variable assignment coinciding with $\pi$ on the elements where it is defined if and only if $\Pi_{\pi}(\phi)$ is satisfiable.

\end{property}

So we have reduced the satisfiability checking problem to the problem of finding a polynomial assignment $\pi$ for which $\Pi_{\pi}(\phi)$ is satisfiable. However, for some satisfiable formulas, the search for such a polynomial assignment is bound to fail. This can be easily seen on the simple initial value problem
\[
  \SApp{X}{0}=1\wedge\forall t\;.\; \geq 0 \Rightarrow \SApp{\dot{X}}{t}=\SApp{X}{t}
  \]
that is satisfiable, but not by any polynomial assignment (the only solution of the given initial value problem is the exponential function).

\subsection{Searching for a Satisfying Polynomial Assignment}
\label{sec:search-algorithm}

Based on the observation that the set of polynomials with rational coefficients is recursively enumerable, for a set of function variables $V$, we denote by $\pi_0^{V},\pi_1^{V},\dots$ an enumeration of all polynomial assignments on $V$. This allows for a trivial way of checking satisfiability of a formula, which is shown in Algorithm~\ref{alg:check}. 
\begin{algorithm}[htb]
  \caption{Search for a Satisfying Polynomial Assignment}
  \label{alg:check}
  \KwIn{formula $\phi$ with function variables from a finite set $V$}
  \DontPrintSemicolon
  $i\leftarrow 0$\\
  \While{$\Pi_{\pi^V_i}(\phi)$ is unsatisfiable}{
    $i\leftarrow i+1$
  }
  \Return{``satisfiable''}
\end{algorithm}

Due to Property~\ref{prop:instantiation}, the algorithm is correct, and if it returns ``satisfiable'', the assignment satisfying the input formula is given by the union of $\pi_i$ (for the function variables) and the assignment satisfying $\Pi^V_{\pi_i}(\phi)$ (for the scalar variables). However, the algorithm only terminates for formulas that are satisfiable by a polynomial assignment. In the rest of this section, we will introduce a natural and practically relevant criterion for this being the case.

\subsection{Robust Satisfiability}
\label{sec:robust-satisf}

Even though differential equations such as $\dot{x}=x$ (in our notation: $\forall t\;.\; \SApp{\dot{X}}{t}=\SApp{X}{t}$) are ubiquitous in mathematics, they are highly idealized objects: In practice, no real physical system will satisfy such an equation precisely, and concrete differential equations can only be used in applications after introducing many simplifying assumptions that are part of the daily bread of practical engineering. However, this also makes it necessary for engineers to assess the consequences of such simplifications. Despite the existence of computability results for ODEs~\cite{Ruohonen:96,Collins:09}, powerful deductive verification techniques~\cite{Platzer:18,Foster:21}, and conservative approximation techniques based on interval arithmetic~\cite{Moore:09,Zgliczynski:02}, in practice, differential equations are still solved by algorithms that produce approximation errors both due to discretization and  due to floating point computation. The reliability of the whole process depends essentially on the fact that the error made by the solver does not dominate the error made by simplifying assumptions. This is a major complication, reliable error analysis of numerical algorithms being a sophisticated extensive research area on its own~\cite{Hairer:87,Suli:03}.

Such complications could be avoided if solvers could conservatively bound the produced errors. For the concrete example $\dot{x}=x$, it would be very useful, if a solver could---instead of solving the differential equation approximately---guarantee the solution of  $x-\epsilon\leq\dot{x}\leq x+\epsilon$ within a compact set, for a small constant $\epsilon>0$.
In this section, we will formally characterize such situations and show that in such cases,  a formally correct satisfiability check is not only possible, but that we can even guarantee its success.

For being able to measure the distance between variable assignments, we will adjoin metrics  to the set of variable assignments $\VarAss$, making the pair $(\VarAss, d)$ a metric space. These metrics will be parametric in a family of compact sets $K_n\subseteq\mathbb{R}^n,n\in\mathbb{N}$ which we will call \emph{domain of interest}. We will denote this dependence on the domain of interest by an index, writing $d_K$ for the metric associated to domain of interest $K$. We will call such a metric on $\VarAss$ a \emph{variable assignment metric}.

\begin{definition}
  \label{def:robsat}
  A formula $\phi$ is \emph{semantically robustly satisfiable} wrt. a variable assignment metric $d$  iff there is a variable assignment $\alpha\in\VarAss$ and an $\varepsilon>0$ (that we call the \emph{robustness margin}) such that for every $\alpha'$ with $d_K(\alpha,\alpha')<\varepsilon$, $\alpha'\models\phi$.
\end{definition}

Note that unlike similar definitions~\cite{Gao:12,Ratschan:02b}, this definition only depends on the semantics of a given formula, but not on its syntax, and hence is invariant wrt. equivalence transformations. We will later see that this is made possible by the fact that we restrict ourselves to operations on real numbers allowed by the decidable theory of real closed fields.

We will usually use metrics induced by some norm, and so we will call a formula robustly satisfiable wrt. a norm~$|| \cdot ||_K$ iff it is robustly satisfiable wrt. the metric $d_K(x,x')= || x-x' ||_K$. Given metrics $d^{\mathcal{T}}$ on $\mathcal{T}$, where  $\mathcal{T}\in \{ \mathcal{R}, \mathcal{F}_1,\dots \}$, we define their extension to variable assignments element-wise. So, for $\alpha,\alpha'\in\VarAss$, 
\[
d_K(\alpha,\alpha'):=
    \max_{\mathcal{T}\in  \{ \mathcal{R}, \mathcal{F}_1,\dots \}}\max_{v\in\mathcal{V}_{\mathcal{T}}} d_K^{\mathcal{T}}(\alpha(v), \alpha'(v)).\]
Here,  we will usually use a family of metrics on function variables of all dimensions. If $d^{\mathcal{R}}$ is a metric on $\mathbb{R}$ and $d^{\mathcal{F}}$ such a family of metrics on smooth functions $\mathbb{R}^i\rightarrow\mathbb{R}, i\in \mathbb{N}$, then we will denote this extension to variable assignments by 
 $d^{\mathcal{R}}\times d^{\mathcal{F}}$.

On real-numbers we will use the discrete metric  \[d^=(x,y):=
   \left\{\begin{array}{ll}
            1, & \text{if } x\neq y,\\
            0, & \text{if } x=y.\\
   \end{array}\right.\]
The metric on functions will be based on a norm measuring the size of a given function and of its derivatives. For a function $F: \mathbb{R}^n\rightarrow\mathbb{R}$, at least $k$-times differentiable,  let \[ || F ||_K^k:= \max_{ |\beta|\leq k} \inf_{x\in K_n} |(D^\beta F)(x)|.  \]
We denote the metric induced by this norm $|| \cdot ||_K^k$ by $d_K^k$. Here are some examples:
\begin{itemize}
\item $\forall t\in [0,1]\;.\; \SApp{X}{t}-0.1\leq \SApp{\dot{X}}{t}\wedge\SApp{\dot{X}}{t}\leq \SApp{X}{t}+0.1$, with $x\in\mathcal{F}_1$ is \emph{not} robustly satisfiable wrt. the norm $||\cdot ||_{[0,1]}^0$, since that norm does not constrain any derivative of $x$. However, it is robustly satisfiable wrt. $||\cdot||_{[0,1]}^1$ since, for example, every function with  maximal distance $0.01$ from the exponential function wrt. $||\cdot||_{[0,1]}^1$ satisfies the formula.

\item $\forall t\in[0,1] \;.\; \SApp{\dot{X}}{t}=\SApp{X}{0}$ is not robustly satisfiable wrt. $||\cdot||^0_{[0,1]}$, since for every function satisfying the formula adding the term $\epsilon t$ to the function results in a function not satisfying it. 
\item The formula $\forall t \;.\; \SApp{X}{t}\geq 0$, while satisfiable, is not robustly satisfiable wrt. the norm $||\cdot ||_{[0,1]}^0$, since this norm only restricts the value of functions in the domain of interest~$[0,1]$. Due to this, for every variable assignment $\alpha$ satisfying the formula, there is an $\alpha'$ with $d(\alpha,\alpha')=0$ that does not satisfy the formula: Simply choose an $\alpha'$ that is identical to $\alpha$ on~$[0,1]$ but reaches a negative value outside of this interval. In contrast to that, the formula $\forall t\in [0,1] \;.\; \SApp{X}{t}\geq 0$ is robustly satisfiable. This explains the importance of bounds on quantified variables for ensuring robustness, which usually is all right in practical applications.
\item The formula $\forall t\in [0,10] \;\exists t'\in [0,10] \;.\; \SApp{X}{t,2t}^2+1\geq \SApp{\partial_2 X}{t',t'}$, which is the result of adding bounds to the first example in this section, is robustly satisfiable wrt. $||\cdot ||_{[0,1]}^1$.
\end{itemize}

\subsection{Robust Completeness}
\label{sec:completeness}

We will now show how the robustness property from Subsection~\ref{sec:robust-satisf} can ensure termination of Algorithm~\ref{alg:check}. We will use the fact that one can approximate smooth functions on compact domains arbitrarily closely by polynomials. For this, recall that a subset $X'$ of a metric space $(X,d)$ is dense in $(X, d)$ iff for every $x\in X$ and $\varepsilon>0$ there is an $x'\in X'$ with $d(x,x')<\varepsilon$. 

\begin{lemma}
  \label{lem:robust}
  Let $\phi$ be a formula, $d$ a variable assignment metric and $\VarAss'\subseteq \VarAss$ s.t. $\VarAss'$ is dense in the metric space $(\VarAss, d)$.   Then every formula that is semantically robustly satisfiable wrt. the metric $d$ has a satisfying assignment from $\VarAss'$.
\end{lemma}

\begin{proof}
  Assume that $\phi$ is semantically robustly satisfiable wrt. $d$. Then there is a variable assignment  $\alpha\in\VarAss$ and an $\varepsilon>0$ such that for every $\alpha'$ with $d(\alpha,\alpha')<\varepsilon$, $\alpha'\models\phi$. Since $\VarAss'$ is dense in $(\VarAss, d)$, there is an $\alpha'\in\VarAss'$ with $d(\alpha,\alpha')<\varepsilon$. Hence $\alpha'$ is within the robustness margin of $\alpha$, and hence $\alpha'\models\phi$.
\Qed\end{proof}

Now we observe:
\begin{lemma}
  \label{lem:ratpoly}
For every $k\in\mathbb{N}_0$, compact $K\subseteq\mathbb{R}^n$, 
the set of $n$-dimensional polynomial functions with rational coefficients is dense in the set of $n$-dimensional polynomial functions with real coefficients wrt. the metric $d_K^k$. 
\end{lemma}

\begin{proof}
  Let $P$ be a polynomial function with real coefficients and $\varepsilon>0$. We will prove that there is a polynomial $P'$ with rational coefficients such that $d_K^k(P, P')<\varepsilon$,  that is $\max_{ |\beta|\leq k} \inf_{x\in K} |(D^\beta P)(x)-(D^\beta P')(x)|<\varepsilon$. During this, we will separate any given polynomial $P$ into its vectors of coefficients $C_P$ and monomials $M_P$. Hence $P=C_P^T M_P$.

  Let $m=\max_{|\beta|\leq k}\inf_{x\in K} ||M_{D^\beta P}(x)||$ with $||\cdot||$ denoting the Euclidean metric. The value $m$ is finite since $K$ is compact. Let $P'$ be a polynomial with rational coefficients s.t. $M_P(x)=M_{P'}(x)$ and 
s.t. $\max_{|\beta|\leq k}||C_{D^\beta P} -C_{D^\beta P'}||<\frac{\varepsilon}{m}$. Now, due to Cauchy-Schwarz,
\begin{multline*}
    \max_{ |\beta|\leq k} \inf_{x\in K} |(D^\beta P)(x)-(D^\beta P')(x)|=
    \max_{ |\beta|\leq k} \inf_{x\in K} |(C_{D^\beta P} -C_{D^\beta P'})^TM_{D^\beta P}(x)|\leq\\
        \max_{ |\beta|\leq k} \inf_{x\in K} ||C_{D^\beta P} -C_{D^\beta P'}||\;||M_{D^\beta P}(x)||<
        \frac{\varepsilon}{m}m=\varepsilon
  \end{multline*}\Qed
\end{proof}

Moreover, the classical Stone-Weierstrass theorem generalizes to $d_K^k$, that is,
the polynomials with real coefficients $P(\mathbb{R})$ are dense in the set of $C_k$-real functions on any compact set~$K$ wrt.  $d_K^k$~\cite{Veretinnikov:16,Narasimhan:85}. %
This allows us to conclude:
\begin{lemma}
    \label{lem:dense}
    For every $k\in\NN_0$ and compact $K\subseteq\RR^n$, the set of $n$-dimensional polynomial functions with rational coefficients is dense in the set of $n$-dimensional smooth functions wrt. the metric $d_K^k$.
\end{lemma}

Now we can apply this to characterize termination of Algorithm~\ref{alg:check}. The algorithm handles scalar variables directly, using a decision procedure for real-closed fields, and hence we need to ensure robustness only wrt. function variables. This amounts to measuring distance on scalar variables using the discrete metric~$d^=$. Hence we get:

\begin{theorem}
  \label{thm:quasi_dec}
  For every $k\in\NN_0$ and family of compact sets $K_n\subseteq\RR^n$, $n\in\mathbb{N}$, Algorithm~\ref{alg:check} terminates for every input formula that is semantically robustly satisfiable wrt. the metric $d^=\times d_K^k$, proving satisfiability of the input formula.
\end{theorem}

\begin{proof}
  Assume a semantically robustly satisfiable function-algebraic formula $\phi$ with function variables from a finite set~$V$.  Since by Lemma~\ref{lem:dense}, the set of variable assignments assigning real values to scalar variables and polynomials with rational coefficients to function variables is dense in the metric space $(\VarAss,d^=\times d_K^k)$, by Lemma~\ref{lem:robust}, there is an assignment $\alpha\in\VarAss$ that assigns polynomials with rational coefficients to function variables and that satisfies the formula~$\phi$, and hence also $\alpha\models\overline{\exists}_{\mathcal{R}}\phi$. Let $\alpha^\pi$ be the polynomial assignment that is the result of restricting $\alpha$ to $V$. Then, due to Property~\ref{prop:instantiation}, $\Pi_{\alpha^\pi}(\phi)$ is satisfiable. Moreover, since the sequence $\pi_0^V, \pi_1^V,\dots$ enumerates the polynomial assignments on $V$, there is an $i$ such that $\pi_i^V=\alpha^\pi$, which ensures that Algorithm~\ref{alg:check} terminates.  \Qed
\end{proof}

Note however, that we do not know how to check a given formula for robustness. Hence, for a given formula we do not know a-priori whether the enumeration algorithm from the proof of Theorem~\ref{thm:quasi_dec} will terminate. We just know that it will terminate \emph{under the assumption} that the formula is robust. Note moreover, that we use the domain of interest $K_n$ only for characterizing the behavior of Algorithm~\ref{alg:check}, but \emph{not} for actual computation. Hence, we also do not need any restriction on the computability of $K_n$.

As an example, consider exponentially stable ordinary differential equations. These have Lyapunov functions that are robust wrt. the metric induced by the norm\footnote{See the proof of Theorem~9 of Peet~\cite{Peet:09}.} 
\[ || F ||_K^k:= \max_{ |\beta|\leq 2} \inf_{x\in K_n} \frac{|(D^\beta F)(x)|}{|x^T x|}.\]
Moreover, the polynomial real-valued functions are dense wrt. this norm~\cite[Theorem~8]{Peet:09}. Even more, an inspection of the proof of this result shows, that indeed, polynomial real-valued functions with \emph{rational} coefficients suffice. Hence our approach turns Peet's purely existential statement~\cite[Theorem~11]{Peet:09} into an algorithmic one:
\begin{corollary}
  Lyapunov functions of exponentially stable ordinary differential equations are computable.
\end{corollary}
Note that exponential stability---as opposed to mere asymptotic stability---is indeed what matters in practice, since asymptotically, but not exponentially stable systems are structurally unstable~\cite{Glas:87}.

\section{General Computable Function Symbols}
\label{sec:transcendental}

\newcommand{\VarAbstr}[1]{\mathit{TAbstr}(#1)}
\newcommand{\VarAbstrAppr}[2]{\mathit{TAbstr}^{#2}(#1)}
\newcommand{\ErrorVar}[1]{\VarAbstrAppr{#1}{\xi}}   %

Now we will remove the restriction to the function-algebraic case. So we assume that in addition to the function symbols $\{ +, \times,\SAppS \} \cup \{ \partial_i\mid i\in\mathbb{N} \}$, the set of function symbols also contains a finite set $\Omega$ of function symbols with signature~$\mathcal{F}_1\rightarrow\mathcal{F}_1$, which we call the \emph{transcendental function symbols}. Moreover, we assume that each such function symbol $\omega\in\Omega$ represents a function such that for every closed interval with rational endpoints $I\subseteq\RR$, for every natural number $k$, one can compute a polynomial $p$ with rational coefficients that approximates the function $\Struct(\omega)$ up to precision $2^{-k}$ in $I$. That is, we require that for all $x$ in $I$, $|p(x)-\Struct(\omega)(x)|\leq 2^{-k}$, which is one of several equivalent definitions of the notion of computable function in computable analysis~\cite{Pour-El:89,Brattka:21}. 
Functions such as the sine function or exponential function are computable in this sense, but also many further functions.

In a similar way to Section~\ref{sec:scal-quant}, we will first provide an algorithm that checks satisfiability of formulas, and will later show that this algorithm is complete under some robustness assumption.

\subsection{Checking Satisfiability}

We start by replacing the transcendental function symbols by fresh function variables.  Denote for any term $t$  by $\VarAbstr{t}$  the result of replacing---from the bottom up---each sub-term of the form $\omega(s)$, where $\omega\in\Omega$, and $s$ an arbitrary term,  by $\SApp{P}{s}$, where---again for each such sub-term---$P$ is a fresh and unique variable of sort $\mathcal{F}_1$. Finally, denote by $\VarAbstr{\phi}$ the result of replacing every term $t$ in $\phi$ that is not a sub-term of another term by  $\VarAbstr{t}$. We will also call the introduced terms of the form $\SApp{P}{t}$ \emph{transabstraction terms}, and the variables $P$ used as their first argument \emph{transabstraction variables}.

\begin{example}
  \label{xpl:var_abstr}
  For the formula $\phi$ being
  \[ \sin(x^2+\cos(y+1))- \SApp{\partial_1 X}{y}^2+10 \leq 0, \]
$\VarAbstr{\phi}$ is
\[ \SApp{P_{\sin}}{x^2+\SApp{P_{\cos}}{y+1}}- \SApp{\partial_1 X}{y}^2+10  \leq 0, \]
where $P_{\sin}$ and $P_{\cos}$ are fresh variables.
\end{example}

We extend variable assignments in such a way that they assign the semantics of the corresponding transcendental function symbols to each transabstraction variable:
\begin{definition}
  For any variable assignment $\alpha$, $\Omega(\alpha)$ is the variable assignment that is equal to $\alpha$ except that for every transabstraction variable $P$ corresponding to a transitive function symbol $\omega\in\Omega$, $\Omega(\alpha)(P):=\Struct(\omega)$.
\end{definition}

Based on this, our way of abstracting transcendental function symbols preserves the semantics of formulas:
\begin{property}
  \label{prop:var_abstr_pres}
  Let $\phi$ be a formula, and $\alpha$ a variable assignment. Then
  \[ \alpha\models \phi \text{ iff } \Omega(\alpha)\models \VarAbstr{\phi}. \]
\end{property}

We will now approximate the transcendental function symbols by polynomials.
 To make the approximation conservative, we bound the approximation error while also ensuring application of the approximating polynomials only within the range where we can ensure that the approximation error stays within this bound. To mimic this approximation error, we inject variables into terms. For this, we denote, for every term $t$ by $\VarAbstrAppr{t}{\xi}$ the result of replacing---again from the bottom up---each sub-term of the form $\omega(s)$ in $t$ by $app(P,s)+\xi$, where $P$ and $s$ are as in the case of $\VarAbstr{t}$, and $\xi$ is another fresh variable, again different for each such term. We call these fresh variables \emph{error variables}. Again, we extend this from terms to formulas in the obvious way.

Based on this, denote
for $\varepsilon\in\mathbb{Q}_{>0}$ by $\VarAbstrAppr{\phi}{\varepsilon}$ the result of 
rewriting every atomic formula $\psi$ in $\phi$ to
\[\forall \xi_{1} \in [-\varepsilon,\varepsilon]\dots \forall \xi_{r} \in [-\varepsilon,\varepsilon]\;.\!\bigwedge_{i\in \{1,\dots, r\}} -1/\varepsilon\leq \VarAbstrAppr{s_i}{\xi}\leq 1/\varepsilon\wedge \VarAbstrAppr{\psi}{\xi},\]
where
\begin{itemize}
\item $s_1,\dots,s_r$ are the terms forming the arguments of the transcendental function symbols in $\psi$, and hence $\VarAbstrAppr{s_i}{\xi}$, $i\in\{1,\dots, r\}$, form the second arguments of the transabstraction terms in $\VarAbstrAppr{\psi}{\xi}$, and
\item $\xi_1,\dots, \xi_r$ are the corresponding error variables.
\end{itemize}

\begin{example}
  Continuing with Example~\ref{xpl:var_abstr}, $\VarAbstrAppr{\phi}{1/10}$ is
  \begin{multline*}
    \forall \xi_{\sin}\in[-1/10,1/10] \;\forall \xi_{\cos}\in[-1/10,1/10]\;.\;\\
    -10\leq x^2+\SApp{P_{\cos}}{y+1}+\xi_{\cos}\leq 10\wedge -10\leq y+1\leq 10\;\wedge\\ \SApp{P_{\sin}}{x^2+\SApp{P_{\cos}}{y+1}+\xi_{\cos}}+\xi_{\sin}  - \SApp{\partial_1 X}{y}^2+10\leq 0.
  \end{multline*}

\end{example}

Any assignment satisfying $\VarAbstrAppr{\phi}{\varepsilon}$ is robust enough to certify satisfaction of $\VarAbstr{\phi}$ when approximating the functions assigned to $P_1,\dots,P_r$ closely enough. This will then allow us to use polynomials as such approximations. For this we will first introduce a tool for measuring the precision of approximating assignments:
\begin{definition}
  For two assignments $\alpha$ and $\alpha'$, set of function variables $\Pi$, and compact set $B\subseteq\mathbb{R}$,
  \[\ d_B^{\Pi}(\alpha,\alpha'):=max_{P\in\Pi} d^0_B(\alpha(P),\alpha'(P)),\] if $\alpha$ and $\alpha'$ coincide on all variables not in $\Pi$, and  $d_B^{\Omega}(\alpha,\alpha')=\infty$, otherwise.
\end{definition}

The following theorem will allow us to infer satisfaction with transcendental function symbols from satisfaction with polynomial functions that approximate the transcendendal function symbols.

\begin{theorem}
\label{thm:conservative}
For every negation-free formula $\phi$, set $\Pi$ of transabstraction variables in $\VarAbstr{\phi}$, for every $\varepsilon>0$, for all assignments $\alpha^P$ and $\alpha^\omega$ such that
  \begin{itemize}
  \item  $\alpha^P\models\VarAbstrAppr{\phi}{\varepsilon}$, and
  \item $d_{[-1/\varepsilon,1/\varepsilon]}^{\Pi}(\alpha^P, \alpha^\omega)<\varepsilon$,

  \end{itemize}
also $\alpha^\omega\models\VarAbstr{\phi}$.
\end{theorem}

\begin{proof}
 Assume an assignment $\alpha^P$ satisfying the preconditions of the theorem.
 We start by considering the case where $\phi$ is an atomic formula. Then
 $\VarAbstrAppr{\phi}{\varepsilon}$ is of the form
 \[\forall \xi_{1} \in [-\varepsilon,\varepsilon]\dots \forall \xi_{r} \in [-\varepsilon,\varepsilon]\;.\; \bigwedge_{i\in \{1,\dots, r\}} -1/\varepsilon\leq \VarAbstrAppr{s_i}{\xi}\leq 1/\varepsilon\wedge \VarAbstrAppr{\phi}{\xi},\]
 with $s_1,\dots,s_r,\xi_1,\dots, \xi_r$ as above.
  We prove that $\alpha^\omega\models \VarAbstr{\phi}$.  For this we construct an extension $\alpha^P_{\xi}$ of $\alpha^P$ s.t.
  \begin{enumerate}
  \item   for all $i\in \{ 1,\dots, r\}$, $\alpha^P_{\xi}(\xi_{i})\in[-\varepsilon,\varepsilon]$, and
  \item $\alpha^P_{\xi}\models \VarAbstrAppr{\phi}{\xi}$ iff $\alpha^\omega\models \VarAbstr{\phi}$.
  \end{enumerate}
  Then, since   $\alpha^P\models\VarAbstrAppr{\phi}{\varepsilon}$, due to the first property,  $\alpha^P_{\xi}\models\VarAbstrAppr{\phi}{\xi}$. From this, the second property ensures $\alpha^\omega\models\VarAbstr{\phi}$, which is what we wanted to prove.
  
  To construct $\alpha_{\xi}^P$, we recursively follow the tree structure of the atomic formula $\phi$ from the leaves to the root, constructing $\alpha_{\xi}^P$ in such a way that for all $i\in\{ 1,\dots, r\}$, 
  \begin{itemize}
  \item $\alpha^P_{\xi}(\xi_{i})\in[-\varepsilon,\varepsilon]$, and
  \item $\alpha^P_{\xi}(\VarAbstrAppr{s_i}{\xi})=\alpha^{\omega}(\VarAbstr{s_i})$.
  \end{itemize}

W.l.o.g. assume that $s_1,\dots,s_{r}$ are ordered in such a way that for every $i,j\in \{1,\dots, {r}\}$ with $i<j$, $s_i$ does not occur below $s_i$ in the syntax tree of the atomic formula~$\phi$.
  
The only difference between  $\VarAbstrAppr{\phi}{\xi}$ and $\VarAbstr{\phi}$ is in the variables~$\xi_i$ added to app-terms. We choose $\alpha^P_{\xi}$ as an extension of $\alpha^P$ such that for all $i\in \{ 1,\dots, {r}\}$, and corresponding transabstraction variables~$\{ P_1,\dots,P_r\}$,
\[\alpha^P_{\xi}(\xi_{i})= \alpha^\omega(\SApp{P_{i}}{\VarAbstr{s_i}})-\alpha^P_{\xi}(\SApp{P_{i}}{\VarAbstrAppr{s_i}{\xi}}),\]
which is well defined, since $\VarAbstrAppr{s_i}{\xi}$ only contains $\xi_{j}$ with $j>i$, and hence 
the value of $\alpha^P_{\xi}(\xi_{_i})$ only depends on values of $\alpha^P_{\xi}(\xi_{j})$ with $j>i$, and especially $\alpha^P_{\xi}(\xi_{_r})$ does not depend on the value of any $\alpha^P_{\xi}(\xi_{i})$, $i\in\{1,\dots,r\}$.

Now we prove that for every $i\in \{1,\dots,r\}$,  $\alpha^P_{\xi}(\VarAbstrAppr{s_i}{\xi})=\alpha^{\omega}(\VarAbstr{s_i})$. First observe that
$\alpha^\omega(\VarAbstr{s_r})=\alpha^P_{\xi}(\VarAbstrAppr{s_r}{\xi})=\alpha^P(\VarAbstrAppr{s_r}{\xi})$, since $\VarAbstr{s_r}=\VarAbstrAppr{s_r}{\xi}$ and $s_r$ neither contains any $P_i$ nor $\xi_i$, $i\in \{1,\dots, {r}\}$.  Now, we prove using induction from ${r}$ down to $1$ that for all $i\in\{1,\dots, {r}\}$, and differing corresponding respective sub-terms $\SApp{P_i}{\VarAbstr{s_i}}$ of $\VarAbstr{\phi}$ and $\SApp{P_i}{\VarAbstrAppr{s_i}{\xi}}+\xi_i$ of $\VarAbstrAppr{\phi}{\xi}$, \[\alpha^P_{\xi}(\SApp{P_i}{\VarAbstrAppr{s_i}{\xi}}+\xi_i)=
\alpha^\omega(\SApp{P_i}{\VarAbstr{s_i}}).\]
This holds due to the following observations:
  \begin{itemize}
  \item $\alpha^P_{\xi}(\VarAbstrAppr{s_i}{\xi})=\alpha^{\omega}(\VarAbstr{s_i})$,
  \item
    $\begin{aligned}[t]
\alpha^P_{\xi}(\xi_i)&=\alpha^\omega(\SApp{P_{i}}{\VarAbstr{s_i}})-\alpha^P_{\xi}(\SApp{P_{i}}{\VarAbstrAppr{s_i}{\xi}})\\
&=\alpha^{\omega}(P_{i})(\alpha^{\omega}(\VarAbstr{s_i}))-\alpha^P(P_{i})(\alpha^P_{\xi}(\VarAbstrAppr{s_i}{\xi})),      
    \end{aligned}$\\
     and since  \[-1/\varepsilon\leq\alpha^P_{\xi}(\VarAbstrAppr{s_i}{\xi})=\alpha^{\omega}(\VarAbstr{s_{i}})\leq 1/\varepsilon,\] also \[|\alpha^{\omega}(P_{i})(\alpha^{\omega}(s_i))-\alpha^P(P_{i})(\alpha^P_{\xi}(\VarAbstrAppr{s_i}{\xi}))|\leq ||\alpha^P(P_\omega)-\alpha^\omega(P_\omega)||^{0}_{[-1/\varepsilon,1/\varepsilon]}\leq \varepsilon,\] and hence $\alpha^P_{\xi}(\xi_{i})\in[-\varepsilon,\varepsilon]$, and finally
   \item $\alpha^P_{\xi}(\SApp{P_{i}}{\VarAbstrAppr{s_i}{\xi}}+\xi_i)
            =\alpha^P_{\xi}(\SApp{P_{i}}{\VarAbstrAppr{s_i}{\xi}})+\alpha^P_{\xi}(\xi_i)=
\alpha^P_{\xi}(\SApp{P_{i}}{\VarAbstrAppr{s_i}{\xi}})+\alpha^\omega(\SApp{P_{i}}{\VarAbstr{s_i}})-\alpha^P_{\xi}(\SApp{P_{i}}{\VarAbstrAppr{s_i}{\xi}})=
           \alpha^\omega(\SApp{P_{i}}{\VarAbstr{s_i}})$.
  \end{itemize}
  
Then $\alpha^P_{\xi}\models\VarAbstrAppr{\phi}{\xi}$, and $\alpha^{\omega}\models\VarAbstr{\phi}$ iff $\alpha^P_{\xi}\models\VarAbstrAppr{\phi}{\xi}$. Hence $\alpha^{\omega}\models\VarAbstr{\phi}$, which is what we wanted to prove.

In the case where $\phi$ is not atomic, the syntax trees of $\VarAbstr{\phi}$ and  $\VarAbstrAppr{\phi}{\varepsilon}$ are identical up to the parts already handled by the atomic case. Then the theorem holds since $\phi$ is negation-free, and the satisfaction relation of all logical symbols except for negation is monotonic wrt. satisfaction of sub-formulas. \Qed
\end{proof}

\begin{example}
  Let $\phi$ be
  \[ \exists x\forall y\in [-1/2,1/2] \;.\; -\sin(y) +\SApp{X}{y}^2-x \leq 0. \]
Then $\VarAbstrAppr{\phi}{1}$ is
  \begin{multline*}
    \exists x\forall y\in [-1/2, 1/2]\:\forall \xi \in [-1,1]\;.\;\\
  -1\leq y\leq 1\wedge-\SApp{P_1}{y}+\SApp{X}{y}^2-x +\xi \leq 0.    
  \end{multline*}

  Let $\alpha(P_1)=\lambda x \:.\: 0$, $\alpha(X)=\lambda x\:.\: 1$.  Then $\alpha(-\SApp{P_1}{y}+\SApp{X}{y}^2-x +\xi)= 1-\alpha(x)+\alpha(\xi)$. Hence $\alpha\models\VarAbstrAppr{\phi}{1}$, which can be seen by choosing the value for the existentially quantified variable~$x$ large enough. Moreover, $||\alpha(P_1)-\sin||^{0}_{[-1,1]}\leq 1$ and the theorem, together with Property~\ref{prop:var_abstr_pres}, ensures that $\{ X\mapsto \lambda x\:.\:1\}\models\phi$.
\end{example}

From Property~\ref{prop:var_abstr_pres} and Theorem~\ref{thm:conservative} we get:
\begin{corollary}
  \label{cor:converse}
  For every negation-free formula $\phi$, set $\Pi$ of transabstraction variables in $\VarAbstr{\phi}$, for every $\varepsilon>0$, for every assignment $\alpha^P$ such that
      \begin{itemize}
      \item  $\alpha^P\models\VarAbstrAppr{\phi}{\varepsilon}$, and

      \item $d_{[-1/\varepsilon,1/\varepsilon]}^{\Pi}(\alpha^P, \Omega(\alpha^P))<\varepsilon$,
 
  \end{itemize}
 also $\alpha^P\models \phi$.
\end{corollary}

Note that Corollary~\ref{cor:converse} only bounds the absolute value of the distance between the transcendental function symbols and their polynomial approximation, but not the distance between any derivative. This is enough, since the transabstraction variables only occur as arguments of transabstraction terms. Hence it is not possible to form expressions where the transcendental function symbols appear as as an argument of function type. Especially one cannot apply differentiation operators to transcendental function symbols.

We will now extend Algorithm~\ref{alg:check} to the case where the input formula $\phi$ contains transcendental function symbols. In this case, the check for satisfiability of $\Pi_{\pi^V_i}(\phi)$ within the algorithm cannot be done directly, due to the presence of transcendental function symbols. For approximating these, we denote by  $\pi_{\Omega}^\varepsilon$ an arbitrary polynomial assignment that assigns to every transabstraction variable $P$ corresponding to a transcendental function symbol $\omega\in\Omega$ a polynomial $p$ with $d^0_{[-1/\varepsilon,1/\varepsilon]}(p,\Struct(\omega))\leq \varepsilon$. We combine the polynomial assignment $\pi^V_i$ with the polynomial assignment~$\pi_\Omega^\varepsilon$, and check for satisfiability of
$\Pi_{\pi^V_i\cup \pi_\Omega^\varepsilon}(\VarAbstrAppr{\phi}{\varepsilon}))$.
This guarantees the existence of an assignment satisfying the preconditions of Corollary~\ref{cor:converse}, and hence satisfiability of $\phi$. However, this needs a choice for $\varepsilon$ that is small enough. Therefore, the algorithm uses dove-tailing, ensuring that for every $i$ the sequence of choices for $\varepsilon$ goes to zero in the limit. The result is Algorithm~\ref{alg:check_trans}.
\begin{algorithm}[htb]
  \caption{Satisfiability with Transcendental Function Symbols}
  \label{alg:check_trans}
  \KwIn{formula $\phi$ with function variables from a finite set $V$}
  \DontPrintSemicolon
  \For{$k\leftarrow 0,\dots$}{
    $\varepsilon\leftarrow 2^{-k}$\\      
    \For{$i\leftarrow 0,\dots, k$}{
      \If{$\Pi_{\pi^V_i\cup \pi_\Omega^\varepsilon}(\VarAbstrAppr{\phi}{\varepsilon}))$ is satisfiable}{
          \Return{``satisfiable''}
  }}}
\end{algorithm}

The algorithm is correct in the following sense.
\begin{theorem}
  \label{thm:correctness}
  If Algorithm~\ref{alg:check_trans} returns ``satisfiable'', the input formula $\phi$ is satisfiable.
\end{theorem}

\begin{proof}
  If the algorithm returns ``satisfiable'', by \Cref{prop:instantiation}, $\VarAbstrAppr{\phi}{\varepsilon}$ is satisfied by a variable assignment $\alpha^P$. Due to this, the first precondition of \Cref{cor:converse} is true. Moreover, 
 $\alpha^P$ coincides with $\pi_\Omega^\varepsilon$ on the transabstraction variables of $\phi$. So, due to the definition of $\pi_\Omega^\varepsilon$, the variable assignment~$\alpha^P$ also satisfies the second precondition of \Cref{cor:converse}. Due to this, its conclusion, $a^P\models\phi$, holds, as well, which proves satisfiability of the input formula~$\phi$.
\Qed\end{proof}

In the case where the algorithm terminates, returning the answer ``satisfiable'', the polynomial assignment
$\pi^V_i$, together with the value $\varepsilon$ may serve as a certificate, allowing for an independent check of the result using the test $\Pi_{\pi^V_i\cup \pi_\Omega^\varepsilon}(\VarAbstrAppr{\phi}{\varepsilon}))$ only, without any search for $\pi^V_i$ and $\varepsilon$. Moreover, the main workhorse of this check is done by a decision procedure for the theory of real-closed fields, for which verified~\cite{Kosaian:23} and proof producing~\cite{McLaughlin:05} versions are available. Complete formal verification or proof production would need to extend this to~\Cref{thm:correctness} (cf. related for the case without function variables~\cite{Platzer:25}).

However,  Algorithm~\ref{alg:check_trans} may fail to terminate. In the next section, we will study when such an approach is guaranteed to terminate.

\subsection{Robust Completeness}
We will now show how robustness can be used also here to ensure termination of Algorithm~\ref{alg:check_trans}.
However, we now need a stronger requirement than the one provided by Definition~\ref{def:robsat}.
As an example, consider the formula $\sin(0)\leq 0$, which does not restrict function variables, at all, and hence is semantically robustly satisfiable, but for which Algorithm~\ref{alg:check_trans} will fail, since it takes into account approximations of the sine function that are slightly bigger than zero at the origin.

Throughout this section we will assume that all atomic formulas are either of the form $t\leq 0$ or of the form $t<0$, and that formulas do not contain negations. Arbitrary formulas can easily be brought into this form by simple equivalence transformations.

The literature already provides several possibilities for defining the needed notion of robustness~\cite{Ratschan:02b,Ratschan:02f,Gao:12}. From several equivalent definitions we use  the simplest one, here.
\begin{definition}
  \label{def:robapp}
  An assignment $\alpha$ \emph{satifies} a formula $\phi$ \emph{$\varepsilon$-robustly}  iff
$\alpha\models \phi_{-\varepsilon}$ where $\phi_{-\varepsilon}$ is the result of replacing the zeros on the right-hand sides of formulas by $-\varepsilon$. In this case we also write $\alpha\models^\varepsilon\phi$. A formula $\phi$ is \emph{robustly satisfiable} iff there is an assignment $\alpha$ and $\varepsilon>0$ such that $\alpha\models^\varepsilon\phi$.
  \end{definition}

\begin{example}
  The formula \[\sin(x)-10 \leq 0,\] is robustly satisfiable, since $\{ x\mapsto 0 \}\models^1\phi$, which is $\{ x\mapsto 0 \}\models\sin x-10 \leq -1$. The formula $\sin(0)\leq 0$ from above, while being satisfiable, is not robustly satisfiable.
\end{example}

However, equalities form a special challenge. 
\begin{example}
  Consider the formula $\sin x=0$. To bring it into the form required here, we rewrite it to the equivalent formula $\sin x\leq 0\wedge -\sin x\leq 0$ which is satisfiable, but not robustly so.
\end{example}

This is a restriction that is common to such approaches~\cite{Ratschan:02f,Gao:12}. Insisting on the ability to prove equalities requires advanced techniques from the mathematical field of topology~\cite{Franek:12}, which is beyond the scope of this paper. Still, it is possible to 
rewrite the formula $\sin x=0$ to $\sin x-\varepsilon\leq0\wedge -\sin x-\varepsilon\leq 0$, for some small positive rational constant~$\varepsilon$. While satisfiability of this formula does not imply satisfiability of the original formula, it proves that a slight relaxation of the formula is satisfiable, which is enough in many applications.

Note that a bounded formula that has an assignment that satisfies it $\varepsilon$-robustly, is semantical robust wrt. the metrics we use in this paper, if the domain of interest is large enough. This is an easy consequence of the fact that terms in  formulas denote functions that are continuous wrt. these metrics.

To prove a converse of Theorem~\ref{thm:conservative}, we define a function that provides a bound on the allowed approximation error for a given set of function variables $\Pi$ that ensures that the value of a given term $t$ and the values of all function applications to elements of $\Pi$ in $t$ do not differ by more than a given value $\delta$. The intuition is  to allow for a half of the error due to propagation from the argument of every function application, and half of the error due to approximation of the function itself.

\begin{definition}
  \label{def:apprprec}
  For every term $t$, set of variables $\Pi$ of sort $\mathcal{F}_1$, variable assignment $\alpha$, and $\delta\in\mathbb{R}_{>0}$, we recursively define $m^{\Pi}_{\alpha}(t, \delta)$ such that
  \begin{itemize}
  \item for any variable or constant $v$, $m^{\Pi}_{\alpha}(v,\delta)=\infty$,
  \item for any app-term $\SApp{P}{s}$ with $P\in\Pi$, \[m^{\Pi}_{\alpha}(\SApp{P}{s},\delta)= \min \{ \delta/2, m^{\Pi}_{\alpha}(s,\min \{\delta/2,\ModCont{\alpha(P)}{\alpha(s)}(\delta/2)\})\},\] and
  \item for $t$ being a term either of the form $f(s_1,\dots,s_k)$, or of the form $\SApp{T}{s_1,\dots,s_k}$, $T\not\in\Pi$ (the latter includes terms with differentiation operators),
    \[m^{\Pi}_{\alpha}(t,\delta)= \min_{i\in \{1,\dots,k \}}m^{\Pi}_{\alpha}(s_i,\ModCont{F}{\alpha(s_1),\dots, \alpha(s_k)}(\delta))\] with $F=
    \left\{\begin{array}{l}
      \Struct(f) \text{ in the former case, and} \\
      \alpha(T), \text{ in the latter.} 
    \end{array}\right.$
  \end{itemize}
Here,
  for a $k$-ary continuous function $F$ and $x\in\mathbb{R}^k$, $\ModCont{F}{x}:\mathbb{R}_{>0}\rightarrow\mathbb{R}_{>0}$ such that for all $\varepsilon\in \mathbb{R}_{>0}$, $x'\in\mathbb{R}^k$, $|x-x'|<\ModCont{F}{x}(\varepsilon)$ implies $|F(x)-F(x')|<\varepsilon$.  \end{definition}

Note that $m^{\Pi}_{\alpha}(t, \delta)\in \mathbb{R}^{>0}\cup\{\infty\}$, and $m^{\Pi}_{\alpha}(t, \delta)$ is either $\infty$ or continuous in the values of the scalar variables in the variable assignment~$\alpha$.
 
The bound indeed ensures convergence of approximation:
\begin{lemma}
  \label{lem:bound_func}
  For every term $t$, set of variables $\Pi$ of sort $\mathcal{F}_1$, positive real numbers $\delta$ and $\varepsilon$, and variable assignments $\alpha$ and $\alpha'$ such that
  \begin{itemize}
  \item $0<\delta\leq 1/\varepsilon$,
  \item $d_{[-1/\varepsilon,1/\varepsilon]}^{\Pi}(\alpha, \alpha')\leq m^{\Pi}_{\alpha}(t,\delta)$,
  \item for all app-terms $\SApp{P}{s}$ in $t$ with $P\in\Pi$, $\alpha(s)\in [-\frac{1}{2\varepsilon},\frac{1}{2\varepsilon}]$,
    \end{itemize}
we have that
  \[|\alpha(t)-\alpha'(t)|\leq\delta.\]
\end{lemma}

\begin{proof}
 We use induction on the structure of terms to prove that the lemma is true for every term. So we prove that the lemma holds for a term $t$ under the assumption that it holds for all of its sub-terms. For doing so by distinguish the three cases of Definition~\ref{def:apprprec}:

 If $t$ is a variable or constant, then 
 $|\alpha(t)-\alpha'(t)|= 0\leq\delta$.  Otherwise, we assume that the lemma holds for all strict sub-terms of $t$ and prove that it holds for $t$, itself. To prove that the lemma holds for $t$, we assume arbitrary, but fixed, $\delta, \varepsilon,\alpha,\alpha'$ such that the pre-conditions of the lemma hold  and prove  that  
 $|\alpha(t)-\alpha'(t)|\leq\delta$.
  
 if $t$ is an app-term $\SApp{p}{s}$ with $p\in\pi$, then
  \begin{align*}
    |\alpha(t)-\alpha'(t)|&=\\
    |\alpha(\SApp{p}{s})-\alpha'(\SApp{p}{s})|&=\\
    |\alpha(p)(\alpha(s))-\alpha'(p)(\alpha'(s))|&
  \end{align*}
To prove that the latter is less or equal $\delta$, we use the pre-condition
\[d_{[-1/\varepsilon,1/\varepsilon]}^{\Pi}(\alpha, \alpha')\leq m^{\Pi}_{\alpha}(\SApp{P}{s},\delta).\]
Due to the definition of $m$, this is
\[d_{[-1/\varepsilon,1/\varepsilon]}^{\Pi}(\alpha, \alpha')\leq \min\{ \delta/2,m^{\Pi}_{\alpha}(s,\min\{\delta/2,\ModCont{\alpha(P)}{\alpha(s)}(\delta/2)\})\},\]
which implies
\[d_{[-1/\varepsilon,1/\varepsilon]}^{\Pi}(\alpha, \alpha')\leq\delta/2,\]
and
\[d_{[-1/\varepsilon,1/\varepsilon]}^{\Pi}(\alpha, \alpha')\leq m^{\Pi}_{\alpha}(s,\min\{\delta/2,\ModCont{\alpha{P}}{\alpha(s)}(\delta/2)\}).\]
The choices $s$ for the universally quantified $t$, and $\min\{\delta/2,\ModCont{\alpha(P)}{\alpha(s)}(\delta/2)\}$ for $\delta$ satisfy the pre-conditions of the induction hypothesis, and we get the conclusion
\[|\alpha(s)-\alpha'(s)|\leq \min\{\delta/2,\ModCont{\alpha(P)}{\alpha(s)}(\delta/2)\}.\]
Due to the definition of $\ModCont{\alpha(P)}{\alpha(s)}$ this implies
\[|\alpha(P)(\alpha(s))-\alpha(P)(\alpha'(s))|\leq\delta/2.\]
This accounts for the difference between $\alpha(s)$ and $\alpha'(s)$ in     $|\alpha(P)(\alpha(s))-\alpha'(P)(\alpha'(s))|$. It remains to bound the error due to the difference between $\alpha(P)$ and $\alpha'(P)$.
The conclusion $|\alpha(s)-\alpha'(s)|\leq \min\{\delta/2,\ModCont{\alpha(P)}{\alpha(s)}(\delta/2)\}$ also implies
\[|\alpha(s)-\alpha'(s)|\leq \delta/2\leq \frac{1}{2\varepsilon}\]
and since $\alpha(s)\in [-\frac{1}{2\varepsilon},\frac{1}{2\varepsilon}]$,
\[\alpha'(s)\in[-1/\varepsilon,1/\varepsilon].\]
From this, and $d_{[-1/\varepsilon,1/\varepsilon]}^{\Pi}(\alpha, \alpha')\leq \delta/2$, we get   \[|\alpha(P)(\alpha'(s))-\alpha'(P)(\alpha'(s))|\leq\delta/2.\]
Now we can combine both errors, and get
\begin{align*}
  |\alpha(P)(\alpha(s))-\alpha'(P)(\alpha'(s))|&=\\
|\alpha(P)(\alpha(s))-\alpha(P)(\alpha'(s))+\alpha(P)(\alpha'(s))-\alpha'(P)(\alpha'(s))|&\leq\\ 
|\alpha(P)(\alpha(s))-\alpha(P)(\alpha'(s))|+|\alpha(P)(\alpha'(s))-\alpha'(P)(\alpha'(s))|&=\\
      \delta/2+\delta/2&=\delta.  
\end{align*}
  If $t$ is of the form $f(s_1,\dots,s_k)$ then we have to prove that under the pre-conditions of the lemma
\[|\alpha(f(s_1,\dots,s_k))-\alpha'(f(s_1,\dots,s_k))|\leq\delta.\]
Due to the definition of $m$, we know that
\[d_{[-1/\varepsilon,1/\varepsilon]}^{\Pi}(\alpha, \alpha')\leq m^{\Pi}_{\alpha}(f(s_1,\dots,s_k),\delta)=
  \min_{i\in \{1,\dots,k \}}m^{\Pi}_{\alpha}(s_i,\ModCont{\Struct(f)}{\alpha(s_1),\dots, \\\alpha(s_k)}(\delta))\]
        and hence, due to the induction hypothesis, for every $i\in\{ 1,\dots, k \}$, \[|\alpha(s_i)-\alpha'(s_i)|\leq\ModCont{\Struct(f)}{\alpha(s_1),\dots, \alpha(s_k)}(\delta).\]
Now we have 
        \begin{align*}
|\alpha(f(s_1,\dots,s_k))-\alpha'(f(s_1,\dots,s_k))|=\\
          |\Struct(f)(\alpha(s_1),\dots,\alpha(s_k)))-\Struct(f)(\alpha'(s_1),\dots,\alpha'(s_k)))|
        \end{align*}
        and due to the definition of $\ModCont{\Struct(f)}{\alpha(s_1),\dots, \alpha(s_k)}(\delta)$,
        \[|\alpha(f(s_1,\dots,s_k))-\alpha'(f(s_1,\dots,s_k))|\leq\delta,\] which is what we wanted to prove.
The case where $t$ is an app-term whose first argument is a term not equal to a variable from $\Pi$ is analogous to the previous case. \Qed
\end{proof}

We extend the previous lemma to the case, where the term is, in addition, perturbed by the values of auxiliary variables. 
\begin{lemma}
  \label{lem:bound}
  For every term $t$ with $r$ occurrences of transcendental function symbols from $\Omega$, and corresponding terms $\SApp{P_i}{s_i}+\xi_i$, $i\in \{1,\dots,r\}$ in $\ErrorVar{t}$, 
 for every real number $\delta>0$ and $\varepsilon>0$, variable assignment $\alpha$, $\alpha'$ such that
  \begin{itemize}
  \item $d_{[-1/\varepsilon,1/\varepsilon]}^{\Pi}(\alpha, \alpha')\leq m^{\Pi}_{\alpha}(t,\delta)/2$, and
  \item for all $i\in\{1,\dots,r\}$,
    $\alpha(s_i)\in [-\frac{1}{2\varepsilon},\frac{1}{2\varepsilon}]$, and
    $|\alpha'(\xi_i)|\leq m^{\Pi}_{\alpha}(s_i,\delta)/2$,
        \end{itemize}
we have that
  \[|\alpha(\VarAbstr{t})-\alpha'(\ErrorVar{t})|\leq \delta.\]

\end{lemma}

\begin{proof}
  Let $\alpha''$ be such that
  \begin{itemize}
  \item for all $i\in\{1,\dots,r\}$, $\alpha''(\xi_i)=0$,
  \item for all $i\in\{1,\dots,r\}$, for all $x\in\mathbb{R}$, $\alpha''(P_{i})(x)=\alpha'(P_{i})(x)+\alpha'(\xi_i)$, and
  \item for all other elements, $\alpha'$ and $\alpha''$ are identical.
  \end{itemize}
  Then \[\alpha'(\ErrorVar{t})=\alpha''(\VarAbstr{t}),\] \[d_{[-1/\varepsilon,1/\varepsilon]}^{\Pi}(\alpha, \alpha'')\leq d_{[-1/\varepsilon,1/\varepsilon]}^{\Pi}(\alpha,\alpha')+d_{[-1/\varepsilon,1/\varepsilon]}^{\Pi}(\alpha', \alpha'')=m^{\Pi}_{\alpha}(t,\delta)/2+ m^{\Pi}_{\alpha}(t,\delta)/2=m^{\Pi}_{\alpha}(t,\delta),\] and Lemma~\ref{lem:bound_func} implies
\[|\alpha(\VarAbstr{t})-\alpha'(\ErrorVar{t})|=|\alpha(\VarAbstr{t})-\alpha''(\VarAbstr{t})|\leq \delta.\]\Qed
\end{proof}

Based on this, we can prove the following converse of Theorem~\ref{thm:conservative}:
\begin{theorem}
  \label{thm:conv}
  Let $\phi$ be a bounded negation-free formula. Let $\Pi$ be the set of transabstraction variables in $\VarAbstr{\phi}$.
  Let $\alpha^{\omega}$ be an assignment and $\delta>0$ s.t. $\alpha^{\omega}\models^{\delta}\VarAbstr{\phi}$. Then there is an $\varepsilon>0$ such that for every $\varepsilon'\leq \varepsilon$, for every $\alpha^P$ with   $d_{[-1/\varepsilon,1/\varepsilon]}^{\Pi}(\alpha^\omega,\alpha^P)<\varepsilon'$,
 $\alpha^P\models\VarAbstrAppr{\phi}{\varepsilon'}$.
\end{theorem}

\begin{proof}
  First will first discuss the case of a single atomic formula $t\leq 0$, with the terms $s_1,\dots, s_r$ and error variables $\xi_1,\dots, \xi_r$ being as above.
  Since $\delta>0$, $m^{\Pi}_{\alpha^{\omega}}(t,\delta)/2$ is well defined. Moreover, it is greater than zero, as already observed. Now choose $\varepsilon$ as \[\min\{ m^{\Pi}_{\alpha^{\omega}}(t,\delta)/2, \min_{i\in\{1,\dots, r\}} \frac{1}{2|\alpha^{\omega}(s_i)|}\},\]
which is again greater than zero. Let $\varepsilon'\leq\varepsilon$ and $\alpha^P$ with  $d_{[-1/\varepsilon,1/\varepsilon]}^{\Pi}(\alpha^\omega,\alpha^P)<\varepsilon'$,
 be arbitrary, but fixed. Let $\alpha^P_\xi$ extend $\alpha^P$ such that $\alpha^P_\xi(\xi_{1}) \in [-\varepsilon',\varepsilon']\dots \alpha^P_\xi(\xi_{r}) \in [-\varepsilon',\varepsilon']$.
To prove $\alpha^P\models\VarAbstrAppr{\phi}{\varepsilon'}$, it suffices to prove that
  \begin{itemize}
  \item $\alpha^P_{\xi}(\ErrorVar{t})\leq 0$, and
  \item for all $i\in \{ 1,\dots,r\}$, $|\alpha^P_{\xi}(\ErrorVar{s_i})|<1/{\varepsilon'}$.
  \end{itemize}
  For proving the first item, we make the following observations:  
    \begin{itemize}
    \item $d_{[-1/\varepsilon,1/\varepsilon]}^{\Pi}( \alpha^\omega,\alpha_\xi^P)<\varepsilon'\leq\varepsilon\leq  m^{\Pi}_{\alpha^\omega}(t,\delta)/2$, 
    \item for all $i\in\{ 1,\dots, r\}$, $\varepsilon\leq \frac{1}{2|\alpha^{\omega}(s_i)|}$, and hence $|\alpha^{\omega}(s_i)|\leq \frac{1}{2\varepsilon}$
    \item 
    for all $i\in\{ 1,\dots, r\}$ , $|\alpha_{\xi}^P(\xi_i)|\leq\varepsilon'\leq\varepsilon\leq  m^{\Pi}_{\alpha^\omega}(t,\delta)/2$
    \end{itemize}
  Hence,  by Lemma~\ref{lem:bound},   
    $|\alpha^{\omega}(t)-\alpha^P_{\xi}(\ErrorVar{t})|\leq\delta$.
    From this,  since $\alpha^{\omega}\models \phi_{-\delta}$, $\alpha^{\omega}(t)\leq -\delta$, and hence $\alpha^P_{\xi}(\ErrorVar{t})\leq 0$.

For proving the second item, let $i\in \{ 1,\dots,r\}$ be arbitrary, but fixed and prove $|\alpha^P_{\xi}(\ErrorVar{s_i})|<1/{\varepsilon'}$. This is a consequence of our observation that  $|\alpha^{\omega}(s_i)|\leq \frac{1}{2\varepsilon}$, and $\varepsilon'<\varepsilon$.

We now proceed to the case where $\phi$ is a general formula. We first observe that the choice of $\varepsilon$ for atomic formulas is  continuous in the quantified variables, since  $m^{\Pi}_{\alpha^{\omega}}(t,\delta)/2$ is. Due to this, and the extreme value theorem, the infimum of this choice over all quantifier bounds and all atomic formulas is greater than zero. So we choose $\varepsilon$ as this infimum.

Now let $\varepsilon'<\varepsilon$, and  $\alpha^P$ with   $d_{[-1/\varepsilon,1/\varepsilon]}^{\Pi}(\alpha^\omega,\alpha^P)<\varepsilon'$, be arbitrary, but fixed. Due to the proof of the atomic case this means that for all atomic formulas $\psi$ in $\phi$, for all choices of quantified variables, if $\alpha^\omega\models\VarAbstr{\psi}$, then $\alpha^P\models\VarAbstrAppr{\psi}{\varepsilon'}$, as well. Since $\phi$ is negation-free, this means that also  $\alpha^P\models\VarAbstrAppr{\phi}{\varepsilon'}$. \Qed

\end{proof}

We get the corresponding analogy of Corollary~\ref{cor:converse}.
\begin{corollary}
  \label{cor:conv}
  Let $\phi$ be a bounded negation-free formula.  Let $\Pi$ be the set of transabstraction variables in $\VarAbstr{\phi}$. Let $\alpha^{\omega}$ be an assignment and $\delta>0$ s.t. $\alpha^{\omega}\models^{\delta}\phi$. Then there is an $\varepsilon>0$ such that for every $\varepsilon'\leq \varepsilon$, for every $\alpha^P$ such that 
 $d_{[-1/\varepsilon,1/\varepsilon]}^{\Pi}(\alpha^P, \Omega(\alpha^\omega))<\varepsilon$, 
 $\alpha^P\models\VarAbstrAppr{\phi}{\varepsilon'}$.
\end{corollary}

And this implies termination of Algorithm~\ref{alg:check_trans}:
\begin{corollary}
Algorithm~\ref{alg:check_trans} terminates for every bounded negation-free input formula that is both robustly satisfiable and semantically robustly satisfiable wrt. the metric $d^=\times d_K^k$.
\end{corollary}

This generalizes previous results~\cite{Ratschan:02f,Gao:12} from the case with only scalar variables to the case of formulas that also contain functional variables. Moreover, while our proofs assume the transcendental function symbols to be computable in the sense of computable analysis~\cite{Pour-El:89,Brattka:21}, the proofs themselves only refer to classical computability theory and mathematics, hence being accessible to readers not familiar with computable analysis.

\section{Combination}
\label{sec:combination}

Up to now, we have discussed algorithms for several classes of formulas. In each step, we generalized the class of formulas that can be handled, while at the same time making the requirements for termination stronger. 
\begin{itemize}
\item Section~\ref{sec:quantifier-free-case}: no robustness requirements
\item Section~\ref{sec:scal-quant}: semantically robustly satisfiable (Definition~\ref{def:robsat})
\item Section~\ref{sec:transcendental}: robustly satisfiable (Definition~\ref{def:robapp})
\end{itemize}
We will now combine this to design algorithms that work for formulas that are conjunctions whose individual parts belong to different classes.
For example, consider the formula
\begin{multline*}
\SApp{X}{0}=1\;\wedge\\ \forall t\in [0,1]\;.\;
  \SApp{X}{t}-0.1\leq \SApp{\dot{X}}{t}\wedge\SApp{\dot{X}}{t}\leq \SApp{X}{t}+0.1,
\end{multline*}
that is function-algebraic, and whose second part we already identified to be semantically robust wrt. the metric $d^1_{[0,1]}$. However, the formula as a whole is not semantically robust wrt. the metric $d^1_{[0,1]}$ due to the equality $\SApp{X}{0}=1$. Nonetheless, this equality is quantifier-free, and hence one might expect that no robustness requirement should be needed for handling this part. In general, we would like to be able to use our algorithms to prove satisfiability of formulas consisting of a non-robust, but quantifier-free part and a part that is robust, but may contain scalar quantifiers.

We will start with an auxiliary lemma that is a direct consequence of the fact that any element of a semi-algebraic set in $\mathbb{R}^n$ is an element of a cell of a cylindrical algebraic decomposition~\cite{Arnon:86,Jirstrand:95} of $\mathbb{R}^n$.
\begin{lemma}
  \label{lem:close_alg}
  Let $(x_1,\dots,x_n)\in\mathbb{R}^n$ be an element of a semi-algebraic set $S\subseteq\mathbb{R}^n$, and let $\varepsilon>0$.  Then there is $(x_1',\dots,x_n')\in S$ such that for all $i\in \{ 1,\dots, n\}$, $|x_i-x_i'|\leq\varepsilon$ and $x_i'$ is real algebraic.
\end{lemma}

We will also need to concurrently approximate and interpolate smooth functions by polynomials with rational coefficients.
\begin{lemma}
  \label{lem:ip_approx}
  For every smooth function $f:\mathbb{R}^n\rightarrow\mathbb{R}$, and real algebraic $(x_1,y_1),\dots, (x_k,y_k)\in \RR^n\times \RR$ there is a polynomial function~$p:\mathbb{R}^n\rightarrow\mathbb{R}$ with rational coefficients such that $p(x_1)=y_1,\dots,p(x_k)=y_k$ and
  \[||f-p||<2\max_{i\in\{ 1,\dots,k\}} |f(x_i)-y_i|.\]
\end{lemma}

\begin{proof}
  Let $d:=\max_{i\in\{ 1,\dots,k\}} |f(x_i)-y_i|$. Let $B$ be a bump function s.t. $B(0)=1$, $\sup_{x\in\mathbb{R}^n}B(x)=1$, and for all $x\in\RR^n$ with $||x||\geq 1$, $B(x)=0$. Let $e:= \min_{i,j\in \{ 1,\dots, k\}, i\neq j} ||x_i-x_j||$. Let \[f'(x)=f(x)+\sum_{i=1,\dots,k} (y_i-f(x_i)) B(\frac{x-x_i}{e}).\] Then for every $i\in \{1,\dots, k\}$, $f'(x_i)=y_i$, and $||f-f'||<d$. But, in general, $f'$ will not be a polynomial.

  Now apply a theorem by Evard and Jafari~\cite[Theorem 6]{Evard:94} to $f'$ to obtain a polynomial~$p$ with
  $||p-f'||<d$ that has identical values to $f'$ at $x_1,\dots, x_k$. Since $||f-f'||<d$ and  $||p-f'||<d$, also $||f-p||<2d$, as required. However, in general, $p$ does not have rational coefficients.

  To ensure rational coefficients, observe that the Evard and Jafari construct the interpolating polynomial from Hermite interpolation, which can be computed by solving a linear system of equations whose coefficients are given by a confluent Vandermonde matrix whose entries are polnomial in the interpolation points. Since the interpolation points are real-algebraic numbers, the resulting Hermite interpolation polynomial has real-algebraic number coefficients, and by transitivity of algebraic dependence there is also an interpolating polynomial with rational coefficients. \Qed
\end{proof}

The problem with the application of Algorithm~\ref{alg:check} to formulas such as the one above lies in the fact that due to the lack of robustness of the first, quantifier-free part we cannot conclude that the formula is satisfiable by an assignment that assigns polynomials with rational coefficients to all functional variables. The following theorem ensures this:

\begin{theorem}
  \label{thm:perturb}
  For every assignment $\alpha$ satisfying a conjunctive formula $\phi$ that is function-algebraic and quantifier-free, for every $\varepsilon>0$ there is a variable assignment $\alpha'$ with $d_K^k(\alpha,\alpha')<\varepsilon$ that also satisfies $\phi$, and that assigns
  polynomials with rational coefficients to all function variables.
\end{theorem}

\begin{proof}
  Let $\alpha$ be an arbitrary, but fixed variable assignment satisfying the formula~$\phi$, and let $\varepsilon>0$.
  Let the formulas $\pi_R$ and $\pi_U$ be the result of applying the variable abstraction phase of the Nelson-Oppen method to $\tau(\phi)$. Hence, $\pi_R$ is a formula in the language of real closed fields, and $\pi_U$ a formula in the language of uninterpreted function symbols.
  Let $I^\tau$ be the $\TheoryRU$-interpretation corresponding to $\alpha$, constructed in  the same way as in the  $\Rightarrow$-part of the proof of Theorem~\ref{thm:equisat}, but also assigning according values to the fresh variables in $\pi_R$ and $\pi_U$ introduced by variables abstraction. Let $T_U$ be the set of sub-terms in $\pi_U$, and let $\sim_\tau\subseteq T_U\times T_U$ be an equivalence relation such that for two terms $s$ and $t$ in $T_U$, $s\sim_\tau t$ iff $I^\tau(s)=I^\tau(t)$. 
Let $E$ be the restriction of $\sim_\tau$ to $vars(\pi_U)$.

Then $I^\tau$ satisfies $\rho(vars(\pi_U), E)$, and hence it also satisfies $\pi_R\wedge \rho(vars(\pi_U), E)$. By Lemma~\ref{lem:close_alg}, there is an interpretation~$I''$ that is identical to $I^\tau$ on all function variables, still satisfies $\pi_R\wedge \rho(vars(\pi_U),E)$, but assigns real-algebraic numbers to all scalar variables such that
 \begin{itemize}
 \item for every scalar variable $v$, $|I^{\tau}(v)-I''(v)|<\varepsilon/2$,
 \item for every term $t\in T_U$,  $|I^{\tau}(t)-I''(t)|<\varepsilon/2$ (this is possible since $I^\tau$ interprets all function symbols by smooth, and hence continuous, functions), and
 \item for every pair of terms $(s,t)\in T_U\times T_U$ with $s\not\sim_\tau t$, $I''(s)\neq I''(t)$.
 \end{itemize}
Due to the last item, for every disequality $s\neq t$ in $\pi_U$, $I''(s)\neq I''(t)$, and hence $I''$ satisfies those disequalities. However, $I''$ assigns values to scalar variables that are different from the values $I^\tau$ assigns, and hence it does not necessarily satisfy the equalities in $\pi_U$.

To rectify this problem, let $r: T_U\rightarrow\mathbb{R}$ such that $r(t):= I''(v)$ iff there is a variable $v$ such that $t\sim_\tau v$ (this value is unique, if there are several such variables), and such that $r(t):=I^\tau(t)$, iff there is no such variable.
Observe that for all $t\in T_U$, $|r(t)-I^\tau(t)|<\varepsilon/2$.

Now use Lemma~\ref{lem:ip_approx} to construct an interpretation $I'$ that assigns the values of $I''$ to all scalar variables, and that assigns  functions to function symbols  such that
\begin{itemize}
\item for every term $t\in T_U$, $I'(t)=r(t)$, and
\item for every function symbol $f$, $||I'(f)-I^\tau(f)||<\varepsilon$.
\end{itemize}
Hence, for each equality $s=t$ in $\pi_U$, $I'(s)=I'(t)$, so $I'$ satisfies these equalities.  Moreover, for every disequality $s\neq t$ in $\pi_U$, $I'(s)\neq I'(t)$. Concluding, $I'$ satisfies $\pi_U$. Now let $\alpha'$ be the variable assignment corresponding to $I'$. We have $d^k_K(\alpha,\alpha')<\varepsilon$, and $\alpha'$ satisfies $\phi$. 
\Qed
\end{proof}

\begin{corollary}
  Let $\phi$ be a quantifier-free conjunctive formula, and let $\psi$ be function-algebraic such that $\phi\wedge\psi$ is satisfiable by an assignment that satisfies the subformula $\psi$ semantically robustly wrt. the metric $d^=\times d_K^k$. Then $\phi\wedge\psi$ is satisfiable by an assignment that assigns polynomials with rational coefficients to all function variables.
\end{corollary}

\begin{proof}
  Let $\alpha$ be the assignment satisfying $\phi\wedge\psi$, and semantically robustly satisfying $\psi$. Let $\varepsilon>0$ be the corresponding robustness margin. Then, by Theorem~\ref{thm:perturb}, there is an assignment $\alpha'$ with $d_K^k(\alpha,\alpha')<\varepsilon$ that still satisfies the subformula $\phi$ but assigns polynomials with rational coefficients to all function variables. Since the subformula~$\psi$ is semantically robustly satisfiable, also $\alpha'\models\psi$, and hence $\alpha'\models\phi\wedge\psi$. \Qed
\end{proof}

This implies termination of Algorithm~\ref{alg:check}:

\begin{corollary}
  Let $\phi$ be a quantifier-free conjunctive formula, and let $\psi$ be function-algebraic such that $\phi\wedge\psi$ is satisfiable by an assignment that satisfies the subformula $\psi$ semantically robustly wrt. the metric $d_K^k$.
  Then  Algorithm~\ref{alg:check} terminates for $\phi\wedge\psi$.
\end{corollary}

And we get an analogous result for Algorithm~\ref{alg:check_trans}:

\begin{corollary}
  Let $\phi$ be a quantifier-free conjunctive formula, and let $\psi$ be bounded and negation-free such that $\phi\wedge\psi$ is satisfiable by an assignment that satisfies the subformula $\psi$ both robustly satisfiable and semantically robustly satisfiable wrt. the metric $d^=\times d_K^k$. Then Algorithm~\ref{alg:check_trans} terminates for $\phi\wedge\psi$.
\end{corollary}

Moreover, since the first part $\phi$ of the formula $\phi\wedge\psi$ in the two corollaries belongs to the decidable class identified in Corollary~\ref{cor:decidable}, Algorithms~\ref{alg:check} and~\ref{alg:check_trans} can be combined with a check for satisfiability of this part, which is able to detect unsatisfiability of $\phi\wedge\psi$ in the case where this is caused by $\phi$ alone being unsatisfiable.

\section{Discussion}
\label{sec:discussion}

Algorithms that prove satisfiability by enumerating polynomials are, of course, hopelessly inefficient in practice. Still, our approach may provide useful practical insight. In practice, problems of the kind studied here are solved in many, often distant areas~\cite{Reddy:19,Prajna:07,Ratschan:17,Ravanbakhsh:18a}. A common approach is to restrict
the set of potential solutions to a fixed class of functions given by a parameterized expression (sometimes also called a \emph{template}), and then searching for values for the parameters such that the result of instantiating the parameters by those values represents a solution to the problem.

There are two main classes of templates that are often used here. The first class are templates given by complex expressions~\cite{Edwards:24}, often called neural networks. The second class are polynomials whose coefficients are parameters which allows many methods to exploit the fact that polynomials are linear in their coefficients~\cite{Ratschan:17}. If the given template polynomial does not represent a solution, one can increase the degree of the polynomial. The resulting loop amounts to an enumeration of all polynomials with real coefficients. Our approach (1) formally justifies such algorithms showing that such a loop must terminate for all robustly satisfiable inputs, and (2) generalizes such algorithms and their termination for robust inputs to all formulas belonging to the language used in this paper.

We demonstrated in Section~\ref{sec:scal-quant} that this implies that Lyapunov functions of asymptotically stable polynomial ordinary differential equations can be found algorithmically. We expect that the methodology developed in this paper allows for the straightforward proof of algorithmic versions of further converse results~\cite{Liu:21,Han:21,Ratschan:18} for certificates of properties of continuous dynamical systems.

\section{Conclusion}
\label{sec:conclusion}

We have developed a  framework for decision procedures for a predicate logical theory formalizing a notion that is central to mathematics, computer science, and many other scientific fields---real-valued functions. Our long-term vision is a framework for automated reasoning about  real-valued functions that results in tools that can be used out-of-the-box in a similar way as decision procedures for common first-order theories in the frame of SMT solvers~\cite{Barrett:18}.

  \paragraph*{Acknowledgments}

  This work was supported by the project 21-09458S of the Czech Science Foundation GA ČR, by the research programme of the Strategy AV21 AI: Artificial Intelligence for Science and Society, and
  institutional support RVO:67985807.

\bibliography{sratscha}
\bibliographystyle{abbrv}

\end{document}